\documentclass{article}
\usepackage[utf8]{inputenc}
\usepackage[margin=1in]{geometry}
\usepackage[numbers]{natbib}
\usepackage{graphicx}
\usepackage{amsmath}
\usepackage{amssymb}
\usepackage{amsthm}

\usepackage{algorithm}
\usepackage{enumerate}
\usepackage{authblk}

\newcommand{\comment}[1]{}

\newcommand{\E}{\mathbb{E}}
\newcommand{\R}{\mathbb{R}}
\newcommand{\N}{\mathbb{N}}
\newcommand{\PP}{\mathbb{P}}
\newcommand{\Zs}{Z_{\sigma_n}}

\newtheorem{theorem}{Theorem}
\newtheorem{proposition}{Proposition}
\newtheorem{ass}{Assumption}
\newtheorem{lemma}{Lemma}
\newtheorem{definition}{Definition}
\newtheorem{remark}{Remark}
\newtheorem{example}{Example}
\newtheorem{corollary}{Corollary}

\title{Optimal design of the Barker proposal and other locally-balanced Metropolis--Hastings algorithms}

\author[1]{Jure Vogrinc\footnote{\texttt{jure.vogrinc@warwick.ac.uk}}}
\author[2]{Samuel Livingstone}
\author[3]{Giacomo Zanella}
\affil[1]{Department of Statistics, University of Warwick}
\affil[2]{Department of Statistical Science, University College London}
\affil[3]{Department of Decision Sciences, BIDSA \& IGIER, Bocconi University}

\date{}

\begin{document}

\maketitle

\begin{abstract}

We study the class of first-order locally-balanced Metropolis--Hastings algorithms introduced in \cite{livingstone2019barker}. To choose a specific algorithm within the class the user must select a balancing function $g:\R \to \R$ satisfying $g(t) = tg(1/t)$, and a noise distribution for the proposal increment.  Popular choices within the class are the Metropolis-adjusted Langevin algorithm and the recently introduced Barker proposal.  We first establish a universal limiting optimal acceptance rate of 57\% and scaling of $n^{-1/3}$ as the dimension $n$ tends to infinity among all members of the class under mild smoothness assumptions on $g$ and when the target distribution for the algorithm is of the product form.  
In particular we obtain an explicit expression for the asymptotic 
efficiency of an arbitrary algorithm in the class, as measured by
expected squared jumping distance.  We then consider how to optimise this expression under various constraints.  We derive an optimal choice of noise distribution for the Barker proposal, optimal choice of balancing function under a Gaussian noise distribution, and optimal choice of first-order locally-balanced algorithm among the entire class, which turns out to depend on the specific target distribution. 
Numerical simulations confirm our theoretical findings and in particular show that a bi-modal choice of noise distribution in the Barker proposal gives rise to a practical algorithm that is consistently more efficient than the original Gaussian version.
\end{abstract}

%\begin{keywords}
%Markov chain Monte Carlo; Metropolis--Hastings; Optimal scaling; Barker proposal; Locally-balanced
%\end{keywords}

\section{Introduction}

Markov chain Monte Carlo algorithms are the workhorse of many contemporary statistical analyses, and an essential part of the modern data science toolkit.  
Despite many advances, however, reliable inference  using Markov chain Monte Carlo can still be a cumbersome task.  
It is common for practitioners to dedicate much effort to making careful algorithm design choices and adjusting algorithmic tuning parameters to ensure that performance is adequate for a given problem.  
Failure to do this can be catastrophic; examples for which a well-designed algorithm performs adequately but a less carefully-chosen alternative does not are ubiquitous (e.g. \cite{sherlock2010random}).

Suitable guidelines on the intelligent design and implementation of Markov chain Monte Carlo methods are therefore important.  They are not always easy to offer, however, the best choice of method can depend on the user and the problem at hand. 
In some contexts, a simpler algorithm with less need for adjustment and for which potential problems are easy to diagnose may be preferable. In others contexts, one may be comfortable with more complex methods, which can perform adequately on a larger class of problems if enough fine tuning is done.

For Metropolis--Hastings algorithms, perhaps the most celebrated guidelines concern the choice of optimal acceptance rate \citep{roberts2001optimal}. Rigorous theoretical justification for certain values tend to be restricted to the case in which dimension tends to infinity and the distribution from which samples are desired has a particular structure (such as a product form), but empirically the same values are known to be appropriate in many other settings \citep{roberts2001optimal}. The apparent lack of dependence of these optimal choices on the target distribution allows particularly simple recommendations to be offered to the user for a given algorithm.

Adaptive Markov chain Monte Carlo methods have also facilitated efficient implementation \citep{andrieu2008tutorial,roberts2009examples}. Users can implement adaptive algorithms in which algorithmic tuning parameters are automatically adjusted towards guideline values, using ideas from stochastic optimisation and controlled Markov chains.  When combined with appropriate theory, adaptive algorithms can therefore allow users to implement their chosen method on a given problem without the need for cumbersome hand-tuning. Such innovations have made it possible to develop popular tailored software packages for users of Markov chain Monte Carlo \citep{rosenthal2007amcmc,carpenter2017stan}.

Not all adaptive algorithms are created equally, however. Empirically it has long been observed that certain approaches are more sensitive to tuning than others \citep{neal2003slice}. In recent work \cite{livingstone2019barker} provided some theoretical justification for this phenomenon, in particular showing that for popular gradient-based approaches such as the Metropolis-adjusted Langevin algorithm and Hamiltonian Monte Carlo spectral gaps decay exponentially quickly to zero as the tuning parameters are perturbed from their optimal values.  By contrast, spectral gaps for the simpler random walk Metropolis decay at a polynomial rate, indicating that the algorithm is much more robust to tuning. This has a compounding effect if the tuning parameters are learned adaptively, as adaptive algorithms typically learn based on past samples from the Markov chain, and if these past samples are very poor as a result of the initial tuning parameters being sub-optimal it can mean that the learning occurs very slowly. The moral of the story is that algorithms can still perform poorly in practice even if an optimally-tuned version would in theory perform well.

These findings present a conundrum, as gradient-based algorithms are considered the state-of-the-art in Markov chain Monte Carlo to sample from continuous and smooth distributions when properly tuned. To explore the phenomenon in more detail \cite{livingstone2019barker} introduce a general class of gradient-based algorithms, termed \emph{first order locally-balanced} Metropolis--Hastings, of which the Metropolis-adjusted Langevin algorithm is a special case. Constructing a member of the class requires a Markov kernel, which can be thought of as the initial noise distribution for the transition, together with a \emph{balancing function}, which must satisfy certain properties (see Section \ref{sec:lb}).  The authors consider different choices from within the class, and in particular construct a method called \emph{the Barker proposal}. This algorithm has spectral gaps that are robust to tuning as in the random walk Metropolis. The authors also establish sufficient conditions for geometric ergodicity and some preliminary results on scaling with dimension, suggesting that relaxation times are $\mathcal{O}(n^{1/3})$, where $n$ is the dimension of the state. 
Empirical results in the paper show that the Barker algorithm pairs extremely well with adaptive learning of tuning parameters, and enables reliable sampling on complex examples in which other gradient-based methods may not, despite being remarkably simple to implement. 
%, and can therefore perform extremely well on complex examples in which other gradient-based methods do not, despite being remarkably simple to implement. 
More discussion and a pedagogical derivation of the Barker algorithm is provided in \cite{hird2020fresh}.

Several unexplored questions remain regarding locally-balanced Metropolis--Hastings algorithms. The initial noise distribution in the Barker algorithm is simply chosen to be Gaussian in \citep{livingstone2019barker}, but no justification besides convenience is given for this choice. It could be that a different choice leads to a more effective algorithm.  Similarly, general guidelines on the optimal acceptance rate for the Barker algorithm are not established. More generally, little discussion is provided on other first-order locally-balanced Metropolis--Hastings. It is natural to wonder whether all members of the class will exhibit $\mathcal{O}(n^{1/3})$ relaxation times, if the Metropolis--adjusted Langevin is the most efficient choice when optimally tuned, and indeed whether such a direct quantitative comparison of methods is possible in general. These questions are of both theoretical and practical interest, as they have direct implications for the optimal design of algorithms.

In this paper we make several new contributions. First we present universal results on the optimal choice of acceptance rate and scaling with dimension of any algorithm within the class of first order locally-balanced Markov processes (under mild regularity conditions). In particular in Section \ref{sec:universal} we show that the 57\% guideline acceptance rate for the Metropolis-adjusted Langevin algorithm also holds for the Barker proposal and several other methods, as does the $\mathcal{O}(n^{1/3})$ scaling with dimension as measured by expected squared jump distance. 
Despite having the same optimal acceptance rate and scaling with dimensionality, however, all such schemes have a different asymptotic efficiency, which we explicitly characterize, enabling for principled and generic optimization of the algorithmic design.
We first consider optimal design of the Barker proposal, in particular with respect to the noise distribution, which is chosen to be Gaussian in \cite{livingstone2019barker}. We find, both theoretically and empirically, that it is in fact beneficial to choose a bi-modal noise distribution for each coordinate, and offer some discussion and practical user guidelines in Section \ref{subsec:optimalbarker}. We then consider the case in which the noise distribution is fixed and an optimal balancing function is chosen in Section \ref{subsec:optimalgaussian}, and the general scenario in which both the noise distribution and balancing function are optimized over in Section \ref{subsec:optimalfirstorder}. Both cases yield surprising results, such as optimality being reached by having positive probability of keeping some coordinates fixed at each iteration.
We conduct numerical experiments to verify the theory in Section \ref{sec:simulations}, and provide a discussion in Section \ref{sec:discussion}.  Our theoretical results build on the recently introduced optimal scaling framework of \cite{vogrinc2021counterexamples,zanella2017dirichlet}. 
One powerful aspect of this approach is the ability to analyze fairly generic schemes without requiring overly case-specific calculations (e.g.\ those related to proposal distributions with Gaussian noise, linear drift, etc.), while still obtaining explicit expressions for the asymptotic performances that can directly be compared among algorithms. 
%One powerful aspect of this approach is the ability to directly compare the asymptotic (in dimension) performance of different algorithms through an explicit problem-dependent expression for the asymptotic efficiency, which can be computed for several different algorithms and then directly compared to choose among them.
This allows characterization of the quantitative interplay between fine-scale properties of the target and proposal distributions (such as moments of the noise, aggressiveness of the balancing function and derivatives of the target) in the resulting asymptotic efficiency, thus enabling precise methodological guidance.

\section{Locally-balanced Markov processes}
\label{sec:lb}

\subsection{General framework}
\label{subsec:framework}

Consider a Markov transition kernel $Q$ defined on a Borel space $(\mathbb{X},\mathcal{F})$. We restrict attention to $\mathbb{X} \subset \mathbb{R}^n$ for some finite $n$. We say $Q$ satisfies the detailed balance equations with respect to a probability measure $\pi$ if 
\begin{equation} \label{eq:db}
\int f(x)h(y) \pi(dx)Q(y,dx) = \int f(x)h(y) \pi(dy)Q(y,dx)
\end{equation}
for any $f,h \in L^2(\pi)$. When $Q$ does not satisfy \eqref{eq:db}, a new kernel can be constructed using the concept of a balancing function. Let $g:[0,\infty) \to [0,\infty)$ be such that $g(0) = 0$ and for $t>0$
\begin{equation} \label{eq:balancingfunction}
    g(t) = tg(1/t),
\end{equation}
and note that by Proposition 1 in \cite{tierney1998note} there exists a symmetric set $\mathcal{R} \times \mathcal{R} \in \mathbb{X} \times \mathbb{X}$ such that the Radon--Nikodym derivative 
\begin{equation} \label{eq:RNderiv}
    t(x,y) = \frac{\pi(dy)Q(y,dx)}{\pi(dx)Q(x,dy)}
\end{equation}
is well-defined and such that $0<t(x,y)<\infty$ if $x,y \in \mathcal{R}$ and $t(x,y) = 0$ otherwise. Then the kernel
\begin{equation} \label{eq:balancedkernel}
    \tilde{\mathcal{P}}(x,dy) = g\left\{ \frac{\pi(dy)Q(y,dx)}{\pi(dx)Q(x,dy)} \right\}Q(x,dy)
\end{equation}
satisfies \eqref{eq:db}. However, the kernel $\tilde{\mathcal{P}}$ is not necessarily Markov. One way of enforcing that \eqref{eq:balancedkernel} integrates to one is to restrict attention to $g \leq 1$, ensuring that $\tilde{\mathcal{P}}(x,\mathbb{X}) \leq 1$, and then combine with $r(x,dy) = \{1 - \tilde{\mathcal{P}}(x,\mathbb{X})\}\delta_x(dy)$, where $\delta_x(A) = 1$ if $x \in A$ and $0$ otherwise. The resulting kernel $\tilde{\mathcal{P}}(x,dy) + r(x,dy)$ is Metropolis--Hastings (e.g. \cite{tierney1998note}).

An alternative strategy introduced in \cite{zanella2020informed,power2019accelerated,livingstone2019barker} is to instead allow any $g$ for which $\mathcal{Z}(x) = \tilde{\mathcal{P}}(x,\mathbb{X})$ is finite, and then set
\begin{equation} \label{eq:balancedkernel_norm}
    \mathcal{P}(x,dy) = \frac{\tilde{\mathcal{P}}(x,dy)}{\mathcal{Z}(x)}.
\end{equation}
Note that $\mathcal{P}$ does not satisfy \eqref{eq:db} in general, in fact $\mathcal{P}$ is invariant with respect to the measure $\mathcal{Z}(x)\pi(dx)$. A $\pi$-invariant Markov jump process can be constructed, however, by introducing a holding time $\mathcal{Z}(x)$ at each state $x$, and then choosing the next state according to $\mathcal{P}$. This construction is called a \emph{locally-balanced Markov process} (see \cite{power2019accelerated,hird2020fresh} for more detail).

\subsection{First order locally-balanced processes}

The function $\mathcal{Z}(x)$ will not be tractable in general, meaning further work is needed to design a sampling algorithm based on a locally-balanced Markov process. One approach is to restrict attention to symmetric $Q$, and $\pi$ absolutely continuous with respect to the Lebesgue measure on $\mathbb{R}^n$ with differentiable Lebesgue density $\pi(x)$. In this case \eqref{eq:RNderiv} reduces to $\pi(y)/\pi(x)$. From this point several natural first order approximations of this ratio can be taken to construct a new more tractable kernel.  It is argued in \cite{hird2020fresh,livingstone2019barker} that a good choice is the component-wise approximation found by letting $Q(x,dy) = \prod_i \sigma^{-1}\mu\{(dy_i - x_i)/\sigma\}$, where $\mu$ is a centred and symmetric distribution on $\mathbb{R}$ and $\sigma > 0$, and setting
$$
\tilde{P}(x,dy) = \prod_{i=1}^n g\left( e^{(y_i-x_i)\partial_i\log\pi(x)} \right) \mu\left(\frac{dy_i-x}{\sigma}\right),
$$
where $\partial_i = \partial/\partial x_i$ and for any event $A$ the set $(A-x_i)/\sigma = \{ z\in\R : x_i +\sigma z \in A\}$, and its Markovian counterpart
%\jure{ (for $A\subseteq \R^n)$
\begin{equation} \label{eq:1storderlb}
    P(x,dy) = \frac{\tilde{P}(x,dy)}{Z(x)}
    %=
    %\frac{1}{\sigma^{-d}Z(x)}\int_A\prod_{i=1}^n g\left( e^{(y_i-x_i)\partial_i\log\pi(x)} \right) \mu\left(\frac{dy_i-x}{\sigma}\right)
    %\\=
    %\frac{1}{Z(x)}\int_{\frac{A-x}{\sigma}}\prod_{i=1}^n g\left( e^{\sigma z_i\partial_i\log\pi(x)} \right) \mu\left(dz_i\right)\,,
\end{equation}
%}
where $Z(x) = \tilde{P}(x,\mathbb{X})$.  With this approximation certain choices of $g$ and each $Q_i$ lead to familiar forms of $P$. Choosing $g(t) = \surd t$ and $\mu$ to be standard Gaussian, for example, leads to the unadjusted Langevin algorithm \citep{roberts1996exponential}.  The class of kernels obtained by \eqref{eq:1storderlb} is much broader, however, and is currently relatively unexplored.

\subsection{The choice of balancing function}

\cite{livingstone2019barker} suggest the choice of balancing function $g(t) = 2t/(1+t)$, as popularised by \cite{barker1965monte} in the context of Metropolis--Hastings.  With this choice a sample from $P$ can be easily drawn in the following manner.  First sample $z_i \sim \mu$ for each $i$, then set $\beta_{x,i} = \partial_i\log\pi(x)$ and flip the sign of each $z_i$ with probability $F(\beta_{x,i}z_i)$, where $F(x) = e^x/(1+e^x)$.  Finally add this to the current coordinate $x_i$.  To construct a $\pi$-invariant Markov chain a Metropolis--Hastings correction is then applied to this Barker proposal.  See Algorithm \ref{alg:barker} for more detail.

%\begin{algo} \label{alg:barker}
%Simulate from the Barker proposal.
%\vspace*{-12pt}
%\begin{tabbing}
    %\qquad Require: current point $x \in \mathbb{X}$ \\
    %\qquad \enspace For $i=1$ to $d$ \\
    %\qquad \qquad Draw $z_i \sim \mu$\\
    %\qquad \qquad Set $\beta_{x,i} \leftarrow \partial_i %\log\pi(x)$ \\
    %\qquad \qquad Draw $u_i \sim %\text{Bernoulli}(\beta_{x,i}z_i)$ and set $b_i \leftarrow% 2u_i - 1$ \\
%    \qquad \qquad Set $y_i \leftarrow x_i + b_iz_i$ \\
%\qquad \enspace Output $y = (y_1,...,y_d)$
%\end{tabbing}
%\end{algo}

\begin{algorithm} 
\caption{Simulate from the Barker proposal.}\label{alg:barker}
%\vspace*{-12pt}
\begin{tabbing}
    \qquad Require: current point $x \in \mathbb{X}$ \\
    \qquad \enspace For $i=1$ to $n$ \\
    \qquad \qquad Draw $z_i \sim \mu$, and set $\beta_{x,i} \leftarrow \partial_i \log\pi(x)$ \\
    \qquad \qquad Set $y_i \leftarrow x_i + z_i$ with probability $F(\beta_{x,i}z_i)$, and $y_i \leftarrow x_i - z_i$ otherwise\\
\qquad \enspace Output $y = (y_1,...,y_d)$
\end{tabbing}
\end{algorithm}

It is natural to wonder how many choices of $g$ can be made. Two other simple possibilities are $\min(1,t)$ and $\max(1,t)$, the latter being recently studied in \cite{choi2020metropolis}.  The below results show that in fact the family of balancing functions is infinitely large.

\begin{proposition} \label{prop:g_to_even}
Let $\mathcal{H} = \{ h : \mathbb{R} \to [0,\infty) ~ : ~ h(x) = h(-x) \}$ be the space of positive even functions. Then for every $h \in \mathcal{H}$, $g_h(t) = t^{1/2} h(\log t)$ is a balancing function. Conversely, for every balancing function $g$, the function $h_g(x) = e^{-x/2}g(e^x)$ is contained in $\mathcal{H}$.
\end{proposition}

%\begin{proposition}
%For any two balancing functions $g_1$ and $g_2$ and any $a \in \mathbb{R}$, $g_a(t) = g_1(t)^a g_2(t)^{1-a}$ is also a balancing function.
%\end{proposition}

The above provides an explicit parametrisation of $g_h$ in terms of a specific $h \in \mathcal{H}$. The function $t^{1/2}$ can also be replaced with any other balancing function to give a different bijection.  The goal, of course, is to find choices of $g$ for which tractable sampling algorithms can be designed.  In Section \ref{sec:optimal} we design new balancing functions of this nature for specific objectives.

\section{A universality result on the optimal acceptance rate and scaling with dimension}
\label{sec:universal}

\subsection{Preliminaries}

The concept of a log-Metropolis--Hastings random variable will be crucial for our analysis of optimal scaling. We recall some key results here, for more detail see Section~3 of \cite{vogrinc2021counterexamples}.

\begin{definition}\label{def:log-MH}
For a probability measure $\pi$, Markov kernel $Q$ on $(\mathbb{X},\mathcal{F})$ and $\mathcal{R}$ as in Section \ref{subsec:framework}, let $X \sim \pi$ and $Y \sim Q(X,\cdot)$. The associated log-Metropolis-Hastings-ratio random variable is
\[
\rho(X,Y)
 = 
\begin{cases}
~\log\left\{\frac{\pi(dy)Q(y,dx)}{\pi(dx)Q(x,dy)}\right\}   &\text{if } (X,Y)\in\mathcal{R}\times\mathcal{R}\,,\\
~0   &\text{otherwise.}
\end{cases}
\]
\end{definition}

Let $\pi:\R \to [0,\infty)$ be a probability density on $\R$ and for any fixed $\sigma >0$ let $Q_\sigma: \R \times \R \to [0,1]$ be a Markov kernel.  We introduce the product measure $\pi_n(dx) = \prod_{i=1}^n \pi(x_i)dx_i$ on $\R^n$ and the product kernel $\mathcal{Q}_n(x,dy) = \prod_{i=1}^n Q_{\sigma_n}(x_i,dy_i)$, where $(\sigma_n)_{n \in \mathbb{N}}$ is a sequence of positive real numbers. The associated log-Metropolis--Hastings random variable is
$$
\rho(X_n,Y_n) = \sum_{i=1}^n \rho_n(X_{n,i},Y_{n,i}),
$$
where $X_n = (X_{n,1},...,X_{n,n}) \sim \pi_n$, $Y_n \sim \mathcal{Q}_n(X_n,\cdot)$, and $\rho_n$ is the log-Metropolis--Hastings random variable associated with $f$ and $Q_{\sigma_n}$.  The following is established in \cite{vogrinc2021counterexamples}.  

\begin{theorem}\label{thm:log-MH}
Assume that there exists a positive sequence $(a_n)_{n \in \mathbb{N}}$ with $\lim_{n\to\infty}a_n = 0$ such that
$$
\lim_{n\to\infty}\E[\rho_n^21_{\rho_n<-a_n}]/\E[\rho_n^2] = 0,
$$ 
where $1$ denotes the indicator function.
%\begin{enumerate}
    %\item $\lim_{n\to\infty}\E[\rho_n^2]= 0$, where the $n$-th expectation is with respect to $\pi(x)Q_{\sigma_n}(x,dy)dx.$
    %\item 
%\end{enumerate}
%Then
%\[
%    \lim_{n\to\infty}\frac{\E[\rho_n]+\frac{1}{2}\E[\rho_n^2]}{\E[\rho_n^2]}
%     = 
%    0\,.
%\]
%Furthermore if $\pi^{(n)}(x^{(n)})=\prod_{i=1}^n\pi(x_i)$ and $Q^{(n)}_n(x^{(n)},y^{(n)})=\prod_{i=1}^n Q_n(x_i,y_i)$ 
If in addition $(\sigma_n)_{n \in \mathbb{N}}$ is chosen such that $\lim_{n\to\infty}n\E[\rho_n^2]=v^2$ for some constant $v>0$, then as $n\to\infty$
\begin{equation} \label{eq:clt}
%\rho_n^{(n)}(X^{(n)},Y^{(n)}) = rf
\sum_{i=1}^n\rho_n(X_{n,i},Y_{n,i})
\Rightarrow
N\left(-\frac{1}{2}v^2,v^2\right)\,.
\end{equation}
\end{theorem}

\begin{remark}
The expectation $n\E[-\rho_n]$ denotes the Kullback--Leibler divergence between the forward and reverse Markov transition kernels, $\pi(dx)Q(x,dy)$ and $\pi(dy)Q(y,dx)$. In fact, it is further shown in \cite{vogrinc2021counterexamples} that under the above assumptions $\lim_{n\to\infty}\E[\rho_n]/\E[\rho_n^2] = -1/2$, meaning that both the mean and variance on the right-hand side of \eqref{eq:clt} can be interpreted in terms of the Kullback--Leibler divergence in the limiting case $n\to\infty$.
\end{remark}

Guaranteeing the first condition, i.e. understanding how fast $\E[\rho_n^2]$ decays, is key for identifying the optimal scaling of a Metropolis-Hastings algorithm.  The other condition is technical and related to the uniform integrability of $n\rho_n^2$ and to the conditions required in the  Lindeberg's version of Central Limit Theorem (see Theorem 4.15 of \cite{kallenberg1997foundations}). It suffices for example, to show that $\rho_n$ has higher moments that vanish faster than $1/n$.

\subsection{The asymptotic acceptance rate for locally-balanced proposals}

We will establish that the above central limit theorem holds for first order locally-balanced Metropolis--Hastings under Assumption \ref{ass1} below, and then consider optimal acceptance rates and dimension dependence in terms of the expected squared jump distance in each coordinate.  We restrict attention to the class of target distributions $\pi_n(x) = \prod_{i=1}^n \exp\{\phi(x_i)\}$, for some $\phi:\mathbb{R} \to \mathbb{R}$, and impose regularity conditions on $\phi$ below.  Let $b(x) = \log\{g(e^x)\}$, and without loss of generality set $g(1)=1$. %Consequently, by differentiating \eqref{eq:balancingfunction} and inserting $t=e^x=1$, we deduce that $x=b(x)-b(-x)$, $g'(1)=b'(0)=1/2$, $b''(0)=1/4+g''(1)$, and $b'''(0)=0$. \sam{Maybe this part can go in an appendix.}

\begin{ass}\label{ass1}
There exist constants $H\in(0,1)$, $\gamma>0$, $\beta\geq 0$, $\epsilon>0$ such that

\begin{enumerate}[(i)]
    \item $\phi \in \mathcal{C}^{3+H}(\R)$ and for $f=\phi''', \phi''\phi', \phi'^3, \phi''|\phi'|^{1+\beta}$ the integrability condition $\int_\R f(x)^{2+\epsilon}(1+|\phi'(x)|^\beta)\pi(x)dx<\infty$ as well as the mixed growth-H\"older condition
    \[
    \left|f(x+\delta)-f(x)\right|
    \leq
    K(x)\max(|\delta|^H,|\delta|^\gamma)
    \]
    are satisfied. Function $K$ is such that $\int_{\R}K(x)^2(1+|\phi'(x)|^\beta)\pi(x)dx<\infty$.
    \item $b \in \mathcal{C}^{3}(\R)$ and $b',b'',b'''$ are all bounded above.
    \item $\int z^2 \mu(dz) =1$, $\int |z|^\xi \mu(dz) <\infty$ for $\xi = \max(6+3\epsilon,2+2H,2+2\gamma)$, and for all $a\in\R$ and some positive $C_\mu>0$
    \[
    \int_{\R}e^{b(az)}|z|^\xi\mu(dz)
    \leq
    C_\mu(1+|a|^\beta)\int_{\R}e^{b(az)}\mu(dz)
    <
    \infty.
    \]
\end{enumerate}
\end{ass}

Part (i) of the above refers to the target distribution, part (ii) to the balancing function and part (iii) to the interplay between them.  Part (i) is straightforwardly satisfied for many statistical models of interest, for example likelihoods from exponential families and suitably smooth priors.  Part (ii) is satisfied by all cases explicitly studied in the paper, such as $g(t) = \surd t$ and $g(t) = 2t/(1+t)$.  Part (iii) highlights the need to control the growth of $g$ and $b$ using the tails of $\mu$. If $g$ is bounded, as in the Barker case, then any $\mu$ with a moment generating function is sufficient for it to be satisfied (and for many targets actually much weaker conditions are required). When $g(t) = \surd t$, which is not bounded above, then stronger conditions on the tails of $\mu$ are needed, such as Gaussian tails.

Part (i) is explicitly weaker than the typical assumptions made in the optimal scaling literature (e.g. \cite{roberts1998optimal}).  A form of part (i) as well as $\int z^6\mu(dz)<\infty$ and $g\in\mathcal{C}^3$ are crucial to the analysis. Part (iii) imposes uniform control (with respect to $x$) of measures $e^{b(\sigma_nz\phi'(x))}(Z_{\sigma_n}(x))^{-1}\mu(z)dz$ in terms of only the measure $\mu(z)dz$. This is required so that the normalising constants $\Zs$ and their second derivatives are well defined. It may be possible to significantly relax parts (ii) or (iii), especially in specific settings, at the expense of strengthening elsewhere.  The following Proposition identifies some simple cases in which part (iii) is satisfied.

\begin{proposition}
\label{prop:ass1sufficient}
Part (iii) of Assumption~\ref{ass1} is satisfied in the following cases.
\begin{enumerate}[(i)]
\item If $\mu$ has a density with compact support, for any $g$. %, by taking $\beta=0$ and $C_\mu=\sup\{|z|^\xi\,;\mu(z)>0\}/2$
\item If $g$ is bounded, non-decreasing and $\int |z|^\xi \mu(dz) <\infty$ for $\xi$ as in Assumption. \ref{ass1}%, by taking $\beta=0$ and $C_\mu=\|g\|_{\infty}\int |z|^\xi \mu(dz)$.
\item If $g$ satisfies part (ii) of Assumption~\ref{ass1} and there exists $\tilde C_\mu,\tilde\beta>0$ such that for all $a\in\R$
\[
\int_{\R}e^{az}|z|^\xi\mu(dz)
\leq \tilde C_\mu(1+|a|^{\tilde\beta})\int_{\R}e^{az}\mu(dz)
<
\infty.
\]

%by taking $\beta=\tilde\beta$ and $C_\mu=\tilde C _\mu\max(1,\|b'\|_\infty^{\tilde{\beta}})$.

\item If $g$ satisfies part (ii) of Assumption~\ref{ass1}, $\mu$ has a density $\mu\in\mathcal{C}^{1}(\R)$ such that $\lim_{z\to\pm\infty}e^{az}\mu(z)=0$ for any $a\in\R$ and there exists constants $p>1$, $A,B>0$  for which
\[
|z|^{p}\mu(z)
\leq
A\mu(z)-Bz\mu'(z).
\]
%by taking $\beta=\lceil\frac{\xi}{p-1}\rceil$ and an appropriate $C_\mu$.
\end{enumerate}
\end{proposition}

In specific examples we typically verify (i), (ii) or (iv) of Proposition~\ref{prop:ass1sufficient}. For instance, choices of the form $\mu(dz) \propto e^{-|z|^p}dz$ for $p\geq 1$ satisfy (iv). Note that a statement analogous to (ii) but for the function $b$ is not valid. Even if $g$ is bounded, $b$ is only bounded from infinity above, not below. In fact, since $b(x)=x+b(-x)$ holds, $b$ can never be bounded. 
These conditions are required to analyze Taylor series remainder terms for the normalising constant. 
It is apparent from Proposition~\ref{prop:ass1sufficient} that less conditions on $\mu$ must be assumed for the Barker proposal, for which $g$ is bounded, compared to the Langevin choice $g(t) = \surd t$.

%\begin{ass}\label{ass2}
%There exist constants $H\in(0,1)$, $\gamma>1$, $\epsilon\in(0,1)$ such that the following assumptions are satisfied
%\begin{enumerate}[i)]
%\item Potential $\phi$ is $\mathcal{C}^{3+H}(\R)$ and for $f=\phi''', \phi''\phi', \phi'^3$ we have $f\in L^{2+\epsilon}(\pi)$ and a mixed growth-H\"older condition is satisfied. That is there exists $K\in L^{2}(\pi)$ such that
%\[
%\left|f(x+\delta)-f(x)\right|
% \leq 
%K(x)\max(|\delta|^H,|\delta|^\gamma)\,.
%\]
%\item Balancing function $g$ is $\mathcal{C}^{3}(\R)$ and $b'''$, for $b(x):=\log(g(e^x))$ is bounded
%\item Symmetric distribution $k$ satisfies $\int_{\R}w^{6+3\epsilon}\max(|w|^{H(2+\epsilon)},|w|^{\gamma(2+\epsilon)})k(w)dw<\infty$.
%\end{enumerate}
%\end{ass}

%\sam{QFJ: I am a bit confused by the $X_i + \sigma_nW_i$ notation and its connection to $W$.  If we set $W_i = (Y_{n,i} - X_{n,i})/\sigma_n$ then in the MALA case $W_i = \sigma_n\phi'(x)/2 + \xi_i$, where $\xi_i \sim \mu$. It looks from the notation that $W = \xi_i$ (or at least they follow the same distribution), rather than $W = W_i$ (which would mean $E[W^4] \neq 3$ in the MALA case, for example). Is that correct?}

\begin{theorem} \label{thm:CLT}
%Let $\rho_n$ be the log-MH ratio random variable of the locally-balanced proposal associated with target $\pi=e^\phi$, balancing function $g$, distribution $\mu$ and variance parameter $\sigma_n$ 
Under Assumption~\ref{ass1} it holds that $\lim_{n\to\infty}\sigma_n^{-6}\E[\rho_n^2]=\theta^2$ for some $\theta\in[0,\infty)$ %(THIS FIRST PART HOLDS FOR ALL DECAY RATES, WHICH IS IMPORTANT).
In addition, if $\theta>0$ and $\sigma_n$ is chosen such that  $\lim_{n\to\infty}n^{1/6}\sigma_n=\ell$, then
\begin{equation}\label{eq:clt_first_order}
\sum_{i=1}^n\rho_n(X_{n,i},Y_{n,i})
\Rightarrow
N\left(-\frac{1}{2}\ell^6\theta^2,\ell^6\theta^2\right).
\end{equation}
Denoting $\mathfrak{g}=g''(1)$, $\mu_4=\int_\R z^4\mu(dz)$, $\mu_6=\int_\R z^6\mu(dz)$ and $A_\phi = \E_\pi[(\phi''')^2]$, $B_\phi = \E_\pi[(\phi'\phi'')^2]$, $C_\phi = \E_\pi[\phi'\phi''\phi''']$ the
constant $\theta^2$ takes the form
\begin{align} \label{eq:theta_general}
\theta^2 =~ &
     \mu_6  \left\{\frac{1}{144}A_\phi + \left(\frac{1}{4}+\mathfrak{g}\right)^2 B_\phi \right.
     %\\&-~
     -\left.\frac{1}{6}\left(\frac{1}{4}+\mathfrak{g}\right) C_\phi \right\}
     \\&+
     \mu_4\left\{\frac{1}{6}\left(\frac{1}{2}+\mathfrak{g}\right) C_\phi \right.
     %\\&-~
     -\left.2\left(\frac{1}{4}+\mathfrak{g}\right)\left(\frac{1}{2}+\mathfrak{g}\right) B_\phi \right\}
     %\\&+
     +\left(\frac{1}{2}+\mathfrak{g}\right)^2 B_\phi. \nonumber
\end{align}
\end{theorem}
Note that the specific choice of the scaling parameter $\sigma_n \propto n^{-1/6}$ in Theorem~\ref{thm:CLT} is the only rate leading to a non-trivial distributional limit for $\sum_{i=1}^n\rho_n(X_{n,i},Y_{n,i})$, despite the fact that $\lim_{n\to\infty}\sigma_n^{-6}\E[\rho_n^2]=\theta^2$ holds for any decay rate. Note also that the expression for $\theta^2$ depends on both the balancing function $g$ and the distribution $\mu$. In Section \ref{sec:optimal} we consider optimal ways to choose $g$ and $\mu$ for certain purposes.  We consider some example choices below.

\begin{example}
In the Langevin case $g(t)=\surd t$ and $\mu$ is standard Gaussian, so that $g''(1)=-1/4$ and $\mu_4=3$, $\mu_6=15$.  Then
\[
\theta^2
=
\frac{5}{48}A_\phi
+
\frac{1}{8}C_\phi
+
\frac{1}{16}B_\phi
\]
which if $\lim_{x\to\pm\infty}e^{\phi(x)}\phi'(x)\phi''(x)^2=0$ can also be written (using integration by parts)
\[
\theta^2
=
\frac{5}{48}\E\left[(\phi''')^2\right]
-
\frac{1}{16}\E\left[(\phi'')^3\right],
\]
a formula that appears in \cite{roberts1998optimal}.
\end{example}

\begin{example}
For the Barker proposal $g(t)=2t/(1+t)$ and $\mu$ can be any centred and symmetric distribution such that $\int z^6\mu(dz)<\infty$. With these choices $g''(1)=-1/2$ and
\begin{equation} \label{eq:theta_barker}
\theta^2
=
\frac{\mu_6}{144}
\left(A_\phi + 6C_\phi + 9B_\phi\right).
\end{equation}
\end{example}

An important consequence of Theorem \ref{thm:CLT}, and in particular of \eqref{eq:clt_first_order}, is a simple expression for the asymptotic acceptance rate for a first order locally-balanced Metropolis--Hastings algorithm (see e.g.\ Proposition 2.4 in \cite{roberts1997weak}).

\begin{corollary} \label{cor:arate}
Setting $\alpha_n(X,Y) = \min\{1, \sum_{i=1}^n\rho_n(X_i,Y_i)\}$, under the conditions of Theorem \ref{thm:CLT}
$$
\lim_{n\to\infty} \E [\alpha_n] = 2\Phi(-\ell^3\theta/2)
$$
where $\Phi$ is the standard Normal cumulative distribution function.
\end{corollary}

\subsection{Optimal acceptance rates}

Given the simplified limiting expression for $\alpha_n$ in Corollary \ref{cor:arate},  we can consider optimal choices of the constant $\ell$ for a fixed $\theta$, leading to an optimal acceptance rate. We consider optimising the expected squared jump distance here, which is well-studied and has a strong justification motivated by diffusion limits in various settings \citep{roberts2001optimal}.

Using the same notation as above denote by $(\mathcal{E}^{g,\mu}_n)_{n \in \mathbb{N}}$ the sequence of expected squared jump distances for the first (or any other) coordinate, defined as
$$
\mathcal{E}^{g,\mu}_n
=
\E\left[(Y_{n,1}-X_{n,1})^2\alpha\left(X_n,Y_n\right)\right],
%\\
%&=
%\E\left[(Y^{(n),g,k}_1-X_1)^2\left(1\wedge e\right)\left(\sum_{i=1}^n\rho_{\sigma_n}(X_i,Y_i^{(n),g,k})\right)\right]\,,
$$
where $X_n\sim \pi_n$ and $Y_n$ is generated from $X_n$ using a first order locally-balanced proposal, defined in \eqref{eq:1storderlb}, with distribution $\mu$, balancing function $g$ and variance parameter $\sigma_n$.
We have the following result.

%\begin{definition}\label{def:ESJD}
%Let $(\sigma_n)_{n \in \mathbb{N}}$ be a positive sequence with $\lim_{n\to\infty} \sigma_n = 0$. Let $X_n \sim \pi_n$ and $Y_n \sim \mathcal{Q}_n(X_n,\cdot)$, density $\mu$ and scaling parameter $\sigma_n$. Denote the expected squared jump distance of the first coordinate of the associated process
%$$
%\mathcal{E}^{g,k}_{\sigma_n}
%=
%\E\left[(Y_{n,1}-X_{n,1})^2\alpha\left(X_n,Y_n\right)\right]
%\\
%&=
%\E\left[(Y^{(n),g,k}_1-X_1)^2\left(1\wedge e\right)\left(\sum_{i=1}^n\rho_{\sigma_n}(X_i,Y_i^{(n),g,k})\right)\right].
%$$
%\end{definition}

\begin{theorem}\label{thm:ESJD}
Let Assumption~\ref{ass1} and Theorem~\ref{thm:CLT} be satisfied for $\phi$, $\mu$ and $g$ and $\theta>0$. Let $(\sigma_n)_{n \in \mathbb{N}}$ be a positive sequence with $\lim_{n\to\infty} \sigma_n = 0$. If either $\lim_{n\to\infty}n^{1/6}\sigma_n=0$ or $\lim_{n\to\infty}n^{1/6}\sigma_n=\infty$ then as $n \to \infty$
\[
n^{1/3}\mathcal{E}^{g,\mu}_n
\to
0.
\]
If $\lim_{n\to\infty} n^{1/6}\sigma_n = \ell$ for some $\ell\in(0,\infty)$, then as $n\to\infty$
\[
n^{1/3}\mathcal{E}^{g,\mu}_n
\to
h(\ell)
=
2\ell^2\Phi(-\ell^3\theta/2)\,,
\]
where $\Phi$ is the standard Normal cumulative distribution function on $\R$.
Furthermore, there exists a unique optimal $\ell^*(=\ell^*(g,\mu))$ that maximizes $h(\ell)$, for which $2\Phi\{-(\ell^*)^3\theta/2\} \approx 0.574$. The corresponding optimal asymptotic efficiency satisfies
\[
h(\ell^*)
=
C_h\theta^{-2/3}\,,
\]
where $C_h \approx 0.652$.
\end{theorem}

The above shows that any first order locally-balanced Metropolis--Hastings algorithm will have the same asymptotic optimal acceptance rate of 0.57, and that algorithmic efficiency as measured by expected squared jump distance will scale as $\mathcal{O}(n^{-1/3})$ for $n\to\infty$. This includes both Barker and Langevin proposals as well as many other possibilities. Theorem~\ref{thm:ESJD} also suggests a route to both comparison and optimal design of first order locally-balanced Metropolis--Hastings algorithms, in the former case by comparing $\theta^2$ for different choices of $\mu$ and $g$, and in the latter by choosing $\mu$ and $g$ so that $\theta^2$ in Theorem~\ref{thm:CLT} is minimized. According to the same theorem, under Assumption~\ref{ass1} the constant $\theta^2$ will depend on $\phi$ through $A_\phi, B_\phi$ and $C_\phi$, on $\mu$ only through $\mu_4$ and $\mu_6$ and on $g$ only through $\mathfrak{g} = g''(1)$. We explore optimal design under different constraints in the next Section.

%We expect the constant $h(\ell)$ to correspond to the speed measure of an appropriate diffusion limit, as shown in \cite{roberts1998optimal} for the Langevin proposal case. 
%\jure{In the Langevin proposal case, the constant $h(\ell)$ corresponds to the speed measure of the appropriate diffusion limit, as is described in \cite{roberts1998optimal}. We conjecture that the same is true for locally balanced proposals in general. I MAY WRITE MORE HERE JUST BEFORE SUBMISSION}
In the Langevin proposal case, the constant $h(\ell)$ was shown to correspond to the speed measure of an overdamped Langevin diffusion limit in \cite{roberts1998optimal}. We conjecture that the same is true for locally balanced proposals in general, but do not prove explicitly diffusion limit results in this paper. 
Proving diffusion limit results for general locally balanced proposals is a non-trivial open problem, as it would require a conditional version of the central limit theorem in \eqref{eq:clt_first_order} that is hard to obtain in such generality.

\begin{example} \label{example:gaussian}
Take the Gaussian target case $\phi(x) = -x^2/2$. Then $\phi'(x) = -x$, $\phi''(x) = -1$ and $\phi'''(x) = 0$, meaning $A_\phi = C_\phi = 0$ and $B_\phi = \E[x^2] = 1$. For Langevin proposals with $g(t) = \surd t$ and $\mu$ taken as Gaussian, the constant $\theta^2$ in \eqref{eq:theta_general} becomes $\theta^2_L = 1/16$, whereas for the Barker choice $g(t) = 2t/(1+t)$ and the same $\mu$ we have 
$\theta_B^2 = \mu_6/16$. The ratio of asymptotic expected squared jump distances is therefore $(\theta_B/\theta_L)^{2/3} = \mu_6^{1/3}$.  Here $\mu_6 = 15$ meaning that Langevin proposals are asymptotically $15^{1/3} \approx 2.47$ times more efficient than Barker proposals with Gaussian noise when optimally tuned. This is consistent with experiments in Section 5.2 of \cite{livingstone2019barker}. 
\end{example}

%C_h=0.651
% values of h(ell^*) I get for Gaussian (with some Monte Carlo noise)
% MALA 1.58  - theory 1.64
% Barker bimodal 1.54
% Barker Gauss 0.68  - theory 0.665
% Barker Rad 1.60
% 3-points 3.75 (this value increases with d)

% values of h(ell^*) I get for Hyperbolic (with some Monte Carlo noise)
% MALA 0.611
% Barker bimodal 1.11
% Barker Gauss 0.51
% Barker Rad 1.18
% 3-points 2.47 (this value increases with d)

%# HYPERBOLIC NUMBERS SHOULD BE
%# MALA/Barker-Gauss = 1.18  - get 1.188    (stable)
%# Barker-Rad / MALA = 2.08   - get 1.93    (growing)
%# Barker-Rad / Barker-Gauss = 2.46     - get 2.29 (growing)
%# conclusion: they seem to match, but radamacher exact asymptotic kicks in later

\begin{example}\label{ex:hyperbolic}
Consider hyperbolic targets, $\phi(x) = (\delta^2 + x^2)^{1/2}$, with $\delta^2 = 0.1$ as in \cite{livingstone2019barker}. Then $A_\phi \approx 12.99$, $B_\phi \approx 0.22$ and $C_\phi \approx 1.68$. The same calculations as above imply that Langevin proposals are 1.18 times more efficient than Barker proposals with Gaussian noise when optimally tuned, which is also consistent with Section 5.2 of \cite{livingstone2019barker}.
\end{example}
%Both the 2.47 and 1.18 theoretical improvements are consistent with the numerical experiments in Section 5.2 of \cite{livingstone2019barker}. 

\section{Optimal choices among the class of locally-balanced algorithms}
\label{sec:optimal}

\subsection{Optimal choice of noise in the Barker algorithm}
\label{subsec:optimalbarker}

In this setting we fix $g(t) = 2t/(1+t)$ and minimize $\theta^2$ with respect to $\mu$, for a given but arbitrary choice of $\phi$.  In this case $\theta^2$ is given by \eqref{eq:theta_barker}, and the only influence of $\mu$ comes from the sixth moment $\mu_6$. The asymptotic expected squared jump distance can therefore be straightforwardly maximised by minimising the sixth moment of $\mu$ subject to the constraint that $\mu_2 = 1$. Note that by Jensen's inequality $\mu_6 \geq \mu_2^3 = 1$, and in fact the lower bound is uniquely attained by choosing $\mu$ to be a Rademacher distribution, such that if $W \sim \mu$ then $W = 1$ with probability $1/2$ and $W=-1$ otherwise.  We state this formally below.

\begin{proposition}
If $g(t) = 2t/(1+t)$ then $\theta^2$ is minimized when $W \sim \mu$ is chosen to take values $+1$ and $-1$ each with probability $1/2$.
\end{proposition}

We can compare the relative efficiency of Barker with Rademacher versus Gaussian noise using \eqref{eq:theta_barker} in a similar manner to Examples \ref{example:gaussian} and \ref{ex:hyperbolic}.
Doing this shows that for any $\phi$ the Rademacher version will be $\mu_6^{1/3} \approx 2.47$ times more efficient than the Gaussian version.  It is particularly convenient that the optimal choice of $\mu$ does not depend in any way on $\phi$ and therefore generic methodological guidance can be provided for the algorithm. 
The comparison with the Langevin proposal is instead target dependent, as exemplified 
%We compare performance to the Langevin case
below.

\begin{example}
When $\phi(x) = -x^2/2$ as in Example \ref{example:gaussian}, the Barker proposal with Rademacher noise will be exactly as efficient as the Langevin proposal. When $\phi(x) = (\delta^2 + x^2)^{1/2}$ with $\delta^2 = 0.1$ as in Example \ref{ex:hyperbolic} then the Rademacher proposal will be 2.08 times more efficient than the Langevin proposal.
\end{example}

We compare these theoretical results with empirical performances in Section \ref{sec:simulations}. 
The Rademacher version of the Barker proposal is not per se a practical sampling algorithm given that the resulting algorithm will not in general produce a $\pi$-irreducible Markov chain.
One simple alternative that we propose is therefore to set $\mu$ to be an evenly-weighted mixture of two Normal distributions centred at $\pm\sqrt{1-\sigma^2}$, each with variance $\sigma^2<1$. The resulting approach, termed bi-modal Barker, will satisfy $\mu_6=1+12\sigma^2+18\sigma^4-16\sigma^6$ and be $15^{1/3}\mu_6^{-1/3}$ times more efficient than the version with Gaussian noise. For small $\sigma$ this is close to optimal whilst also being practical. For instance, for the choice $\sigma^2=0.1^2$, which is the one we use in simulations below, bi-modal Barker is approximately $2.37$ times more efficient than the Gaussian version.

The result on the Radamacher optimality may seem surprising at first given given the lack of $\pi$-irreducibility. Similar results have, however, been uncovered previously, for example it is known that the optimum expected squared jump distance for the random walk Metropolis when the target distribution is spherically symmetric is found by choosing the proposal distribution to be uniform on a hyper-sphere of fixed radius from the current point \citep{neal2011optimal}. Given the product form of $\pi$ considered in this work, the Rademacher structure is therefore natural. For the random walk Metropolis, however, the benefits of choosing such an optimized proposal distribution vanish as the dimension increases \citep{neal2011optimal,yang2013searching}, whereas in the Barker case they do not.

An intuitive explanation for this may be that bi-modal Barker proposal choice makes the MCMC method less diffusive and puts more effort on moving at least a certain distance away. This is consistent with motivation for other kinds of development of MCMC methods, for instance Hamiltonian Monte Carlo and non-reversible Piece-wise deterministic Markov processes \citep{duane1987hybrid,fearnhead2018piecewise}.

\subsection{Optimising over the choice of balancing function for a fixed noise distribution}
\label{subsec:optimalgaussian}

In this section we switch attention to the optimal choice of $g$ for a fixed choice of $\mu$.  The expression \eqref{eq:theta_general} in this case becomes a simple quadratic in $\mathfrak{g}$, which can be straightforwardly solved to find an optimum choice for a given $\phi$, as given in \eqref{eq:gopt_fixedmu} below.

\begin{proposition}\label{prop:opt_g}
Given $\phi$ and a fixed noise distribution $\mu$ with finite fourth and sixth moments $\mu_4 < \mu_6 < \infty$, the optimum choice of $\mathfrak{g}$ is %the minimum of the quadratic equation
%\begin{align*}
%\left( \mu_6 B_\phi + 2\mu_4 B_\phi + B_\phi \right)&\mathfrak{g}^2
%+ 
%\left\{ \mu_6(B_\phi/2 - C_\phi/6) + \mu_4(c_\phi/6 - 3B_\phi/2) + B %\right\}\mathfrak{g} \\
%+& 
%\mu_6(A_\phi/144 + B_\phi/16 - C_\phi/24) + \mu_4(C_\phi/12 - %B_\phi/4) + B_\phi/4,
%\end{align*}
%which is
\begin{equation} \label{eq:gopt_fixedmu}
\mathfrak{g}^* = \frac{\mu_6 \left(C_\phi - 3B_\phi \right) + \mu_4\left( 9B_\phi - C_\phi \right) - 6B_\phi}{12B_\phi\left( \mu_6 - 2\mu_4 + 1\right)}.
\end{equation}
\end{proposition}

Any family of balancing functions for which $\mathfrak{g} = g''(1)$ can be modified to take a desired value could therefore in principle be used to create an optimized algorithm for a particular $\mu$ and $\phi$. 
Consider the family
\begin{equation}\label{eq:g_family}
g_\gamma(t) = \frac{1}{2}\left(t^{\frac{1}{2}+\gamma}+t^{\frac{1}{2}-\gamma}\right),
\end{equation}
indexed by $\gamma \geq 0$, where for $\gamma=0$ we recover the Langevin case $g(t) = \surd t$. 
%By design, $g_\gamma$ is a balancing function with $g_\gamma(1)=1$ and $\mathfrak{g}=g_\gamma''(1)=\gamma^2-\frac{1}{4}$, and
Any choice within the family is a balancing function, and is such that $g_\gamma(1)=1$ and $\mathfrak{g}=g_\gamma''(1)=\gamma^2-\frac{1}{4}$.  For a given $\phi$, the choice of $\gamma$ can therefore be adjusted to achieve the optimum asymptotic efficiency provided that $\mathfrak{g}^*$ in \eqref{eq:gopt_fixedmu} is larger than $-1/4$.
%\jure{This is especially efficient when $C_\phi/(10B_\phi)-1/5>-\frac{1}{4}$ when the best possible $\mathfrak{g}$ can be achieved.}

%Given the results of the previous section it would seem natural to set $\mu$ as a Rademacher distribution. However in this case Proposition \ref{prop:opt_g} does not apply since $\mu_4=\mu_6$, and it turns out that all choices of $g$ give equivalent algorithms. 
Given the results of the previous section it would seem natural to set $\mu$ as a Rademacher distribution, however in this case it turns out that all choices of $g$ give equivalent algorithms.
This follows straightforwardly from the fact that \eqref{eq:balancingfunction} implies $g(t)/\{g(t) + g(t^{-1})\} = 1/(1+t^{-1})$, which is independent of $g$. In fact Proposition \ref{prop:opt_g} does not apply 
to the Rademacher case %when $\mu$ is the Rademacher distribution 
since $\mu_4=\mu_6$. 
Another natural option is to fix $\mu$ to be standard Gaussian.
In this case 
%Setting $\mu_6 = 15$ and $\mu_4 = 3$ into
\eqref{eq:gopt_fixedmu} implies that the maximum efficiency is found by choosing
$\mathfrak{g}=C_\phi/(10B_\phi)-1/5$. 
%Interestingly, neither the Langevin nor Barker $g$ are optimal there since the optimum will instead be target dependent.
This scheme can be implemented using the family in \eqref{eq:g_family}, and sampling from the resulting first order locally-balanced proposal 
%with using \eqref{eq:g_family} can be done
is straightforward as it consists in a mixture of two Gaussians, see the supplement for details. 
%Sampling from a first order locally-balanced proposal  with using \eqref{eq:g_family} can be done straightforwardly for many choices of $\mu$ including the Gaussian, see Appendix \ref{app:sampling} for details.
We do not implement this scheme in the simulations, however, in favour of the more efficient alternatives discussed in the next section. %, where we optimize $\mu$ and $g$ jointly.

\subsection{Optimising over the choice of both noise distribution and balancing function}
\label{subsec:optimalfirstorder}

In this section we consider optimizing over both $g$ and $\mu$ jointly. %We can again without loss of generality restrict attention to the value $\mathfrak{g} = g''(1)$ and the fourth and sixth moments of $\mu$, which we will again denote by $\mu_4$ and $\mu_6$.
The following proposition identifies the best possibly achievable asymptotic efficiency with first order locally-balanced proposals for a given target.

\begin{proposition}
A non-negative lower bound for $\theta^2$ that is independent of both $\mu$ and $g$ is
\begin{equation} \label{eq:ESJDglobalmin}
\theta^2 \geq \frac{1}{144}\left( A_\phi - \frac{C_\phi^2}{B_\phi} \right).
\end{equation}
Furthermore, $\theta^2$ can be made arbitrarily close to the lower bound by choosing $\mu_4>1$ sufficiently close to one, setting $\mu_6 = \mu_4^2$ and choosing
\begin{equation}\label{eq:g_opt}
\mathfrak{g}= \frac{\mu_4(C_\phi-3B_\phi)  + 6B_\phi}{12B_\phi(\mu_4-1)}.
\end{equation} 
\end{proposition}

\begin{proof}
Given $A_\phi>0, B_\phi>0$ and $C_\phi\in \R$ we must solve the constrained quadratic optimisation problem of minimising $\theta^2$ subject to $1 \leq \mu_4 \leq \surd \mu_6$.  The constraints on $\mu_4$ and $\mu_6$ are necessary because $1= \mu_2 \leq \surd \mu_4$ by Jensen's inequality and $\mu_4\leq \surd (\mu_2\mu_6) = \surd \mu_6$ by Cauchy's inequality. Moreover the Hamburger moment problem tells us these constraints are sufficient: if they are fulfilled then there exists a symmetric proposal distribution on $\R$ that satisfies them.

Defining the new variables $m_1 = (\mathfrak{g}+1/4)\phi'\phi'' - \phi'''/12$ and $m_2 = -(\mathfrak{g}+1/2) \phi'\phi''$, we can rewrite %and find a lower bound for 
$\theta^2$ as
\begin{equation}\label{eq:first_ineq}
        \theta^2 
    = 
    \left(\mu_6-\mu_4^2\right)\E[m_1^2] + \E\left[\left(\mu_4m_1+m_2\right)^2\right] 
    \geq 
    \E\left[\left(\mu_4m_1 + m_2\right)^2\right].
\end{equation}
where the inequality follows from $\surd \mu_6 \geq \mu_4$.
Expressing this lower bound in terms of $A_\phi,B_\phi$ and $C_\phi$ gives
$$
\theta^2 \geq B_\phi\left\{\left(\mu_4\mathfrak{g}-\mathfrak{g} + \frac{\mu_4}{4}-\frac12\right)-\frac{\mu_4}{12}\frac{C_\phi}{B_\phi}\right\}^2
 + \frac{\mu_4^2}{144}\left(A_\phi-\frac{C_\phi^2}{B_\phi}\right),
$$
which can itself be lower bounded, giving
$$
\theta^2 \geq \frac{\mu_4^2}{144}\left(A_\phi-\frac{C_\phi^2}{B_\phi}\right)
\geq 
\frac{1}{144}\left(A_\phi-\frac{C_\phi^2}{B_\phi}\right).
$$
We have used three inequalities. The first, in \eqref{eq:first_ineq}, %inequality holds because $\surd \mu_6 \geq \mu_4$ and an equality 
is realised if and only if $\mu_6=\mu_4^2$; the second simply bounds a square below by zero and is realised if and only if $\mathfrak{g}$ is defined as in \eqref{eq:g_opt}, which requires $\mu_4>1$; the third relies on $\mu_4\geq 1$ and is realised if and only if $\mu_4 = 1$. Note that the last two equalities cannot be realised simultaneously. 
The final lower bound is always non-negative due to $B_\phi A_\phi\geq C_\phi^2$ by Cauchy's inequality.
\end{proof}

Denote by $\nu(a)$ for $a>1$ a discrete symmetric distribution taking three possible values $-\sqrt{a},0,\sqrt{a}$, such that the probability of a non-zero value is $1/a$, and note that this is the unique symmetric distribution $\mu$ with moments satisfying $\mu_2=1$, $\mu_4=a$ and $\mu_6=a^2$.
Letting $\mathfrak{g}$ be defined by \eqref{eq:g_opt}, choosing $\mu=\nu(\mu_4)$ and taking $\mu_4$ arbitrarily close to one results in $\theta^2$ becoming arbitrarily close to the lower bound \eqref{eq:ESJDglobalmin}. This three point proposal results in an algorithm that achieves close to optimal asymptotic expected squared jump distance among the class of first order locally-balanced samplers provided that $\mathfrak{g}$ is chosen according to \eqref{eq:g_opt}.

\begin{remark}
This three point proposal is in fact also the optimal choice of $\mu$ for any fixed choice of $g$, but the amount of mass given to point zero  will vary depending on $g$. In the Barker case, for example, this point achieves no mass, resulting in the Rademacher choice for $\mu$.
\end{remark}

It is natural to consider taking the limit $\mu_4 \to 1$ and expect optimality to be reached there. When the dimension $n$ is fixed and finite, however, this results in a Rademacher proposal, which is suboptimal.  This can be seen by noting that the lower bound \eqref{eq:ESJDglobalmin} is always smaller than $B_\phi/16+A_\phi/144+C_\phi/24$, the value attained by the Rademacher proposal, because
$$
\frac{B_\phi}{16}+\frac{A_\phi}{144}+\frac{C_\phi}{24}
=
\frac{A_\phi}{144}+\left(\frac{B_\phi^{1/2}}{4}+\frac{C_\phi}{12B_\phi^{1/2}}\right)^2-\frac{C_\phi^2}{144B_\phi}
\geq \frac{1}{144}\left(A_\phi-\frac{C_\phi^2}{B_\phi}\right).
$$
Inspecting the proof of Theorem~\ref{thm:CLT} shows that $\mathfrak{g}$ must be increased sufficiently slowly as a function of $n$ to control the remainder terms in order for the asymptotic expression for $\theta^2$ to be a valid representation of the expected squared jump distance. In other words, as $\mu_4\to 1$ it takes increasingly large $n$ for the asymptotic regime to be representative of the finite $n$ setting. For a finite $n$, it is therefore necessary to choose $\mu_4>1$. We explore this phenomenon further in the supplement.
In all simulations below, we set $\mu_4=2$ unless stated otherwise. %, as empirically we found values of $\mu_4$ roughly between $1.5$ and $3$ to provide good performances across different scenarios.

A surprising consequence of these findings is that the three point proposal with some mass at zero outperforms a Rademacher choice that is optimum for the Barker proposal when the freedom to choose $\mathfrak{g}$ is given. In terms of sampling, this suggests that efficiency gains can be made by allowing some components of the state to remain unchanged at each iteration of the algorithm with a probability that depends on the size of the gradient in that direction. The same family of balancing functions introduced in \eqref{eq:g_family} can again be used to create this optimum sampler.

A particular case of interest is the Gaussian setting $\phi(x) = -x^2/2$, in which case 
$\phi'''(0)$ and therefore $A_\phi=C_\phi=0$.
This means that by choosing any $\mu_4>1$ and $\mathfrak{g}$ according to \eqref{eq:g_opt} we can achieve zero asymptotic $\theta^2$.
 The result of this is a super-efficient sampler whose efficiency will effectively decay at a slower rate than $n^{-1/3}$. We illustrate this surprising finding numerically in Section \ref{sec:simulations}, but also stress that this property only holds when $\phi(x) = -x^2/2$ to the best of our knowledge.

\section{Simulation Study}
\label{sec:simulations}

\subsection{Efficiency with dimension on product targets}\label{sec:scaling_product}

We examine the expected squared jump distance of the first component of two different product form target distributions as a function of dimension. This setting is directly captured by the theoretical results of Sections \ref{sec:universal} and \ref{sec:optimal}.  The two target distributions considered are the multi-dimensional standard Gaussian distribution and the hyperbolic distribution of Example \ref{ex:hyperbolic}.  In each case we compare the random walk Metropolis, the Metropolis-adjusted Langevin algorithm, Barker with Gaussian noise, Barker with Rademacher noise, Barker with bi-modal noise as described in Section \ref{subsec:optimalbarker} and the optimal choice over both balancing function and noise distribution described in Section \ref{subsec:optimalfirstorder}, which will hereafter be called the three point proposal.

The results for the Gaussian target distribution are shown in Figure \ref{fig:esjd_prod}(a). It is clear from the plots that among the Barker algorithms the Rademacher and bi-modal choices are comparable and perform similarly to MALA, whereas the Barker algorithm with Gaussian noise has a lower expected squared jump distance by a factor of 2-2.5, in accordance with the theoretical value of 2.47.  The three points proposal performs best and appears to exhibit a slightly slower than $n^{-1/3}$ decay in expected squared jumping distance when the dimension in large.  This is because in the special case of Gaussian target $\theta^2$ from \eqref{eq:theta_general} equals %can be made arbitrarily close to 
zero when the choices described in Section \ref{subsec:optimalfirstorder} are made.

For the hyperbolic target results are shown in Figure \ref{fig:esjd_prod}(b).  The main difference compared to the Gaussian example is that now the Barker algorithms with Rademacher and bi-modal noise both outperform the Langevin algorithm, as predicted by the theory described in Section \ref{subsec:optimalbarker}.  The three points proposal is still the best performing algorithm.

\begin{figure}[h!]
\centering
\includegraphics[width=0.49\linewidth]{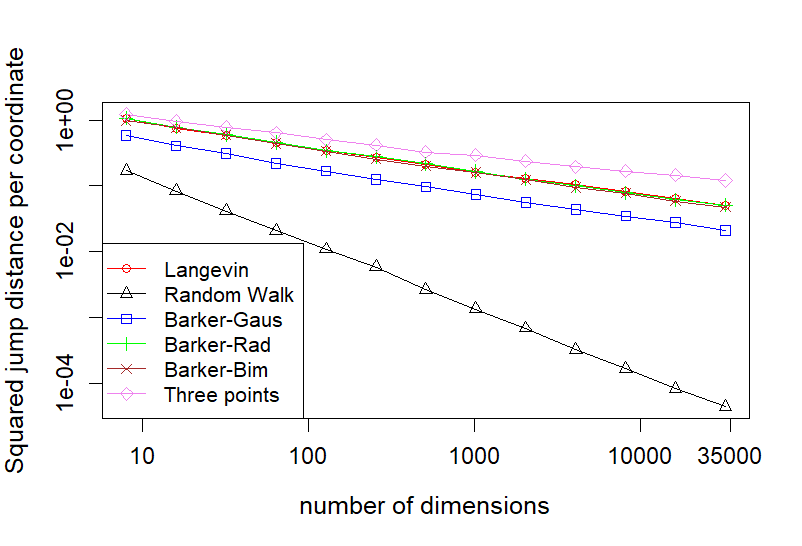}
\includegraphics[width=0.49\linewidth]{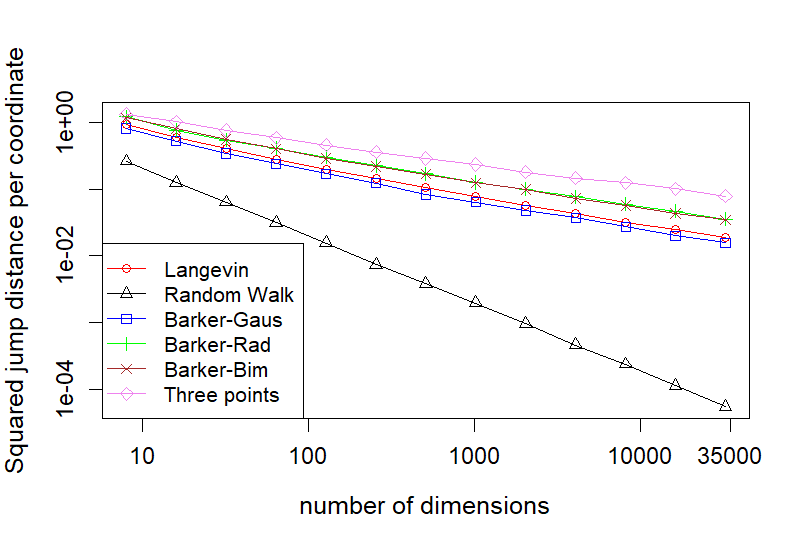}
\caption{ESJD against dimensionality. Left: Gaussian product target. Right: Hyperbolic product target.
}
\label{fig:esjd_prod}
\end{figure}

\subsection{Poisson random effects model}

To consider a realistic example in which the target distribution is not of the product form, we compare algorithms on the Poisson random effects model described in Section 6.3 of \cite{livingstone2019barker}. 
We compare the Barker algorithm with bi-modal noise to the Barker algorithm with Gaussian noise, the Langevin algorithm and the random walk Metropolis.
The main purpose of this example is to assess whether or not the above theoretical guidelines for the noise distribution in the Barker algorithm lead to good choices even when the target distribution does not have independent and identically distributed components.%is no longer of the product form. 
%The main purpose of this example is to assess whether or not the theoretically guidelines for the noise distribution in the Barker algorithm for suitably regular product form targets is still a sensible choice when the target distribution is no longer of the product form and instead exhibits dependence. 

The target distribution under consideration is a $51$-dimensional posterior distribution,  $p(\mu,\eta_1,\dots,\eta_{50}|\textbf{y})$, arising from 
a Poisson random effects model defined hierarchically as $\mu\sim \hbox{N}(0,10^2)$, 
$\eta_i|\mu\sim \hbox{N}(\mu,\sigma_{\eta}^2)$ and 
$y_{ij}|\eta_i\sim \hbox{Poisson}(\exp(\eta_i))$, 
independently for 
$i=1,\dots,50$ and
$j=1,\dots,5$. 
%Here  denote the observed data, which 
In our experiment we generate the observed data $\textbf{y}=(y_{ij})_{ij}$ %that we condition on to obtain the posterior 
from the model likelihood, i.e.\ sampling $y_{ij}\sim\hbox{Poisson}(\exp(\eta^*_i))$ independently, where $\eta^*_1,\dots,\eta^*_I$ are themselves generated independently from a $\hbox{N}(\mu^*,\sigma_{\eta}^2)$ distribution with $\mu^*=5$.
%assuming the data-generating value of $\mu$ to be $\mu^*=5$ and sampling the data-generating values of $\eta_1,\dots,\eta_I$ from their prior distribution.
%in this case are $250$ observations, and 
Here $\sigma_{\eta}$ is a fixed value and two scenarios are considered: in the first we set $\sigma_{\eta}=1$, while in the second we set $\sigma_{\eta}=3$. 
Effectively, $\sigma_{\eta}$ is a parameter that governs the heterogeneity across groups $i=1,\dots,50$ in the hierarchy. Thus, larger values of $\sigma_{\eta}$ lead to a target distribution with more heterogeneity of scales across coordinates, which make the adaptation and sampling process more challenging.

 In each case algorithmic tuning parameters consisting of a diagonal pre-conditioning matrix and a global scale are learned using Algorithm 4 of \cite{andrieu2008tutorial}, in the same manner as described in Section 6.3 of \cite{livingstone2019barker}.  
 We measure efficiency in terms of effective sample size for a given number of iterations since all algorithms under comparison, apart from Random Walk Metropolis, have a roughly equivalent cost per iteration, which is dominated by gradient computations.
 Figure \ref{fig:poisson_model} reports the median effective sample sizes across parameters for 100 independent runs of $5\times 10^4$ iterations of each algorithm. 
 All algorithms were randomly initialized by sampling parameter values from their prior distributions.
\begin{figure}[h!]
\centering
\includegraphics[width=0.49\linewidth]{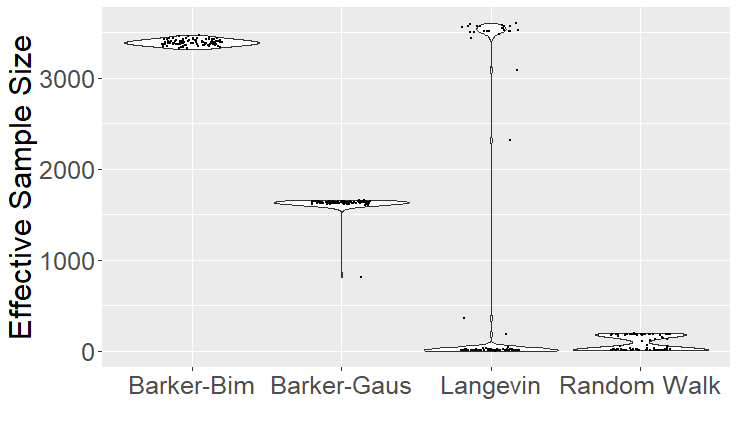}
\includegraphics[width=0.49\linewidth]{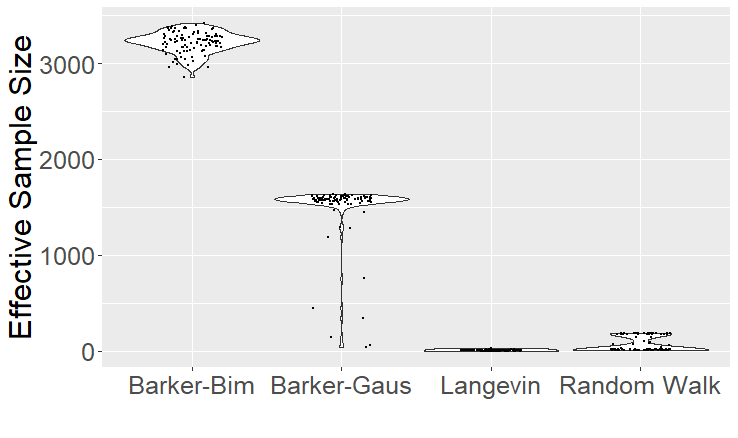}
\caption{Violin plots of median effective sample sizes (median across parameters) for 100 independent repetitions of each algorithm. Left: low heterogeneity across coordinates. Right: high heterogeneity across coordinates.
}
\label{fig:poisson_model}
\end{figure}

Both versions of the Barker algorithm appear to be more robust to different hyperparameter values than the Langevin algorithm, which sometimes performs well but sometimes poorly  in the first scenario and always performs poorly in the second. This is because the Langevin algorithm is very sensitive to tuning parameter selection, and the adaptive procedure fails to converge on sensible values for these across the time scales of the simulation.  The random walk Metropolis also performs poorly, which is largely explained by the dimension of the problem. The Barker algorithm with bi-modal noise is approximately two times as efficient in terms of effective sample size as the version with Gaussian noise in this setting.
More precisely, the median improvement in estimated effective sample size  %for bimodal Barker versus Gaussian Barker 
is 2.08 in scenario 1 (10th and 90th quantiles across the 100 repetitions 2.05 and 2.11 respectively) and 2.04 in scenario 2 (10th and 90th quantiles 1.98 and 2.14 respectively). 
Similar numbers were obtained when looking at minimum (rather than median) effective sample sizes across parameters.
These values suggest that the asymptotic theory developed in this paper, which quantifies bi-modal Barker to be 2.37 times more efficient than Gaussian Barker, is highly predictive of behaviours observed in practice also for moderate dimensionality and targets that have neither independent nor identically distributed coordinates. 
More generally, in all our simulations, we consistently observed a improvement in efficiency when going from Gaussian to bimodal Barker with factors typically between 2 and 2.5.

\subsection{A correlated example}
Unlike the Random Walk or Langevin algorithms, the Barker and three points schemes rely on a choice of coordinate system.
This may raise the concern of how much performance depends on specific choices of coordinate systems, and in particular whether the $\mathcal{O}(n^{1/3})$ scaling behaviour proved above is sensitive to the theoretical assumption that the target factorizes across the same coordinate axes as the proposal. 
%being somehow good coordinates for a specific target under consideration.
Here we explore these issues numerically, performing high dimensional scaling experiments similar to Section \ref{sec:scaling_product} but for non-product form targets with significant correlation. In particular, we consider Gaussian distributions with non-diagonal covariance matrix $\Sigma$ chosen in two ways. In the first case we set  $\Sigma_{ii}=1$ for $i=1,\dots,n$ and 
$\Sigma_{ij}=\rho$ for $i\neq j$, while in the second we take $\Sigma_{ij}=\rho^{|i-j|}$. In both cases we set $\rho=0.99$ to depart drastically from the independence case. 
As in Section \ref{sec:scaling_product} we compute the expected square jump distance per coordinate.
For all algorithms under consideration we use isotropic proposals, meaning we do not use preconditioning to avoid aligning proposal and target axes, and we choose a step-size that is numerically optimized to maximize performance as measured by expected square jump distance.
The results are reported in Figure \ref{fig:scaling_correlated}.
As expected, all schemes perform worse than in the product case (note the different scales on the $y$-axes between Figure \ref{fig:esjd_prod} and Figure \ref{fig:scaling_correlated}), but the relative comparison between different schemes remains nearly unchanged and fully coherent with the theoretical predictions obtained from Sections \ref{sec:universal} and \ref{sec:optimal}. 
In particular the Langevin, Barker bi-modal and Barker Radamacher schemes perform nearly equivalently, while Barker with Gaussian noise performs around 2-2.5 times worse. %between Langevin and the various Barker versions as well as among Barker versions, , and in particular between Langevin and the various Barker versions as well as among Barker versions, 
Overall, the experiment suggests that the relative performances of the Random Walk, Langevin and Barker algorithms is not particularly sensitive to correlation and to the specific choice of coordinate system. 
\begin{figure}[h!]
\centering
\includegraphics[width=0.49\linewidth]{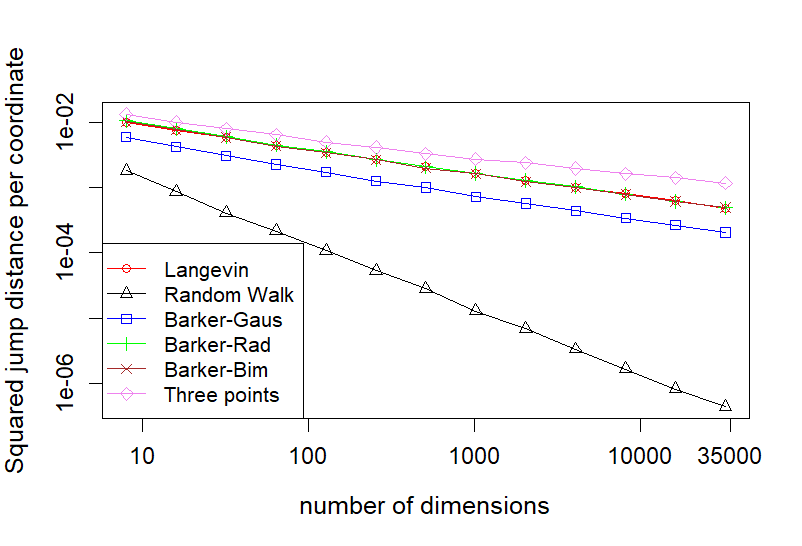}
\includegraphics[width=0.49\linewidth]{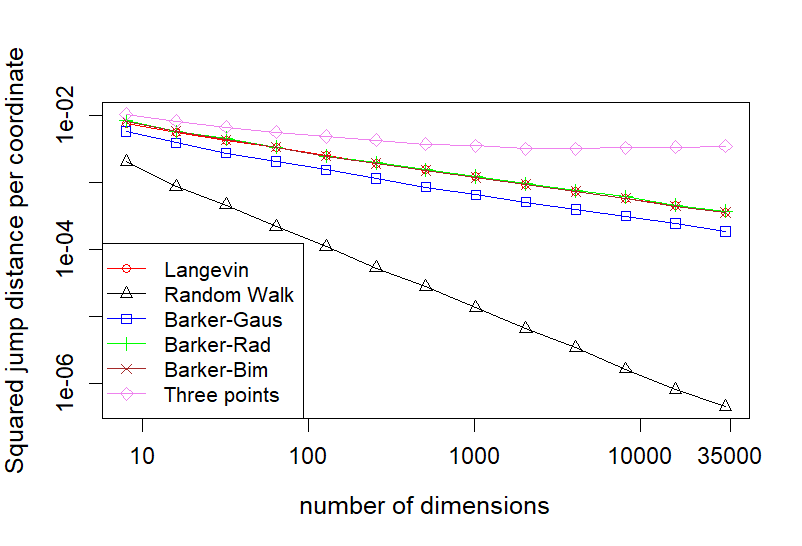}
\caption{Expected squared jump distance against dimensionality for correlated Gaussian targets. 
Left: $\Sigma_{ii}=1$ and $\Sigma_{ij}=0.99$ for $i\neq j$.
Right: $\Sigma_{ij}=0.99^{|i-j|}$.
%AR(1) Gaussian target with $\rho=0.99$.
%In both cases the three point is optimized as if the target was iid Gaussian and $\mu_4=2$.
}
\label{fig:scaling_correlated}
\end{figure}
The three point proposal performs well also in these correlated examples and actually performs surprisingly well when $\Sigma_{ij}=\rho^{|i-j|}$. 
Providing better understanding of such unexpected behaviour will be the subject of future research. 
%Providing a better understanding of the behavior of the three point proposal will be the focus of future research.
Note, however, that the three point proposal implicitly uses knowledge about the target distribution when choosing the optimal values of the tuning parameters $\mathfrak{g}$ and $\mu_4$, and thus it has been given a somewhat unfair and potentially unrealistic advantage compared to the other schemes considered here.
In particular, in this example $\mathfrak{g}$ was chosen according to the optimal value in \eqref{eq:g_opt} with $B_\phi=1$ and $C_\phi=0$ as given by product-form Gaussian targets. %, while the fourth moment was set to $\mu_4=2$ as in the other simulations above. %, see the supplement for sensitivity to this parameter.

\section{Discussion}
\label{sec:discussion}

The main results of this paper rely on a product form structure of $\pi$, and the corresponding optimal choice of locally-balanced algorithm also has a product form.  We have shown in Section \ref{sec:simulations} that this choice is still effective when the target distribution is no longer of the product form, and therefore recommend the use of the bi-modal Barker algorithm in practice.  It is surprising that using a non-local noise distribution of this kind results in such a pronounced and consistent improvement in efficiency across multiple examples.  We believe that this represents a good case study of theoretical analysis motivating new practical methodology that would be otherwise be hard to devise. 
%We also expect, however, that different theoretically optimal choices can be found under different assumptions on the target distribution $\pi$. 
%A spherical symmetry assumption, for example, may lead to spherically symmetric optimal choices of locally-balanced proposal, such as those considered in the case of the random walk Metropolis in Section 3 of \cite{neal2011optimal}. 
%\giac{Which assumptions on the proposal are we thinking at here? It's not clear to me whether a spherically symmetric proposal should be optimal in cases where $g$ does not factorize (or at least it is not clear to me why). Also, I would not want to implicitly convey the message that a product-form proposal is good because the target is product-form (that would be self-harming for this paper, and I actually do not think it is true).}
%We leave such explorations for future work.
It is also worth noting that any improvement in efficiency discussed above essentially comes for free, since 
all the gradient-based schemes considered in the paper have a comparable cost per iteration, which is typically dominated by gradient computations, and all schemes are equally simple to implement.

The detailed quantitative analysis and comparison of algorithms within the locally-balanced class in the high-dimensional limit is made possible by the mathematical framework developed in Section~3 of \cite{vogrinc2021counterexamples}. This framework identifies and uses only essential Taylor series expansions related to the limiting Kullback--Leibler divergence between a locally-balanced proposal and its time reversal.
Using this we establish optimal scaling for a broad class of algorithms including Barker and Langevin with a single unified proof, along with significantly weaker assumptions on the smoothness and tails of the target distribution than those in \cite{roberts1998optimal}.
Our results are at present restricted to limiting expected squared jump distances, rather than diffusion limits as in \cite{roberts1997weak} or \cite{roberts1998optimal}, but we believe that it is possible to uncover a limiting process under the current assumptions and such a line of enquiry is being pursued at the time of writing, continuing the axiomatic approach in \cite{vogrinc2021counterexamples}.

One intriguing finding of this work concerns the sub-optimality of the Langevin choice $g(t) = \surd t$ with proposal input noise $\mu$ chosen to be Gaussian.  This is by far the most historically popular choice within the first order locally-balanced class of algorithms. The results in this paper show that according to asymptotic efficiency as measured by expected squared jump distance not only is this combination of $\mu$ and $g$ not optimum, but in addition that the optimum choice of $\mu$ when $g(t) = \surd t$ is not Gaussian, and also that the optimum choice of $g$ when using Gaussian $\mu$ is not $\surd t$.

A natural open question is whether the insights of Section \ref{subsec:optimalfirstorder} can be used to create a novel new algorithm based on the three point proposal scheme.  We have resisted doing so here because such an algorithm would require a problem-specific choice of balancing function and some appropriate randomisation of the noise distribution to prevent reducibility issues.  It may be possible, however, to design an adaptive Markov chain Monte Carlo method that is able to learn these quantities during the simulation.  We look forward to designing practical methodology based on the insights of Section \ref{subsec:optimalfirstorder} in subsequent work.

%This is the concluding part of the paper.  It is only needed if it contains new material.
%It  should not repeat the summary or reiterate the contents of the paper.

%\begin{itemize}

%\item{Recommend bi-modal Barker in practice}

%\item \jure{Discuss (maybe in previous section) different implementations of the near optimal method from 4.3 (and 4.2). Upgraded algorithm 2 for mixture of Gaussians vs irreducible algorithm with 3-point proposal with randomised $\sigma$ or an occasional Barker step.} \sam{I think we can mention this in the simulations on in the supplement}

%\item \jure{ON THE PROOF TECHNICALITITES: 
%}

%\item\jure{Discuss relations between target-dependent optimisation of an MCMC method with relation to adaptive MCMC}

%\item\jure{Talk about HMC like generalisations of locally-balanced?} \sam{I would rather not do this here as I think we've got a lot already and are close to the page limit.}
%\end{itemize}

\section*{Acknowledgement} JV was supported by a UK Engineering and Physical Sciences Research Council grants EP/R022100/1 and EP/T004134/1.
SL is supported by a UK Engineering and Physical Sciences Research Council grant EP/V055380/1.

%\section*{Supplementary material}
%\label{SM}
%Supplementary material available at \Bka\ online includes further simulations related to the three point proposal of Section \ref{subsec:optimalfirstorder}, illustrating a finite-dimensional example for which it is optimal to choose $\mu_4 > 1$.

\appendix

\section{Sampling from optimal
%the family of 
locally-balanced kernels %given by \eqref{eq:g_family}.
with Gaussian noise} \label{app:sampling}

We consider here fixing $\mu$ to be standard Gaussian, as is the choice made in default versions of the Langevin and Barker algorithms.  Setting $\mu_6 = 15$ and $\mu_4 = 3$ into \eqref{eq:gopt_fixedmu} implies that the maximum efficiency is found by choosing
%\begin{equation} \label{eq:theta_gaussian}
%10B_\phi \mathfrak{g}^2+\mathfrak{g}(4B_\phi-2C_\phi)+\left(\frac{7B_\phi}{16}+\frac{5A_\phi}{48}-\frac{3C_\phi}{8}\right),
%\end{equation}
$\mathfrak{g}=C_\phi/(10B_\phi)-1/5$. %We evaluate this choice in Section \ref{sec:simulations}. \giacomo{(I don't think we do)}
Such value of $\mathfrak{g}$ can be imposed using, e.g., the family of balancing functions defined in \eqref{eq:g_family}.

Sampling from a first order locally-balanced proposal using \eqref{eq:g_family} and standard Gaussian $\mu$ can then be done using Algorithm \ref{alg:g_family} applied to each coordinate.  
Also, any choice of balancing function within the family \eqref{eq:g_family} and $\mu$ chosen as a mixture of Gaussians results in a proposal density with analytically tractable normalising constant, meaning Metropolis--Hastings acceptance rates can be evaluated.

\begin{algorithm}
\caption{Simulate from the locally-balanced proposal using \eqref{eq:g_family}.}\label{alg:g_family}
%\vspace*{-12pt}
\begin{tabbing}
    \qquad Require: $x\in\R$ and $\sigma > 0$ \\
    %\qquad \enspace For $i=1$ to $d$ \\
    \qquad Set $p_x \leftarrow (1+e^{-\gamma\sigma^2\phi'(x)^2})^{-1}$\\
    \qquad Draw $u_x \sim \text{Bernoulli}(p_x)$ and set $b_x \leftarrow 2u_x - 1$\\
    \qquad Draw $z \sim N(0,1)$ and set $y \leftarrow x + \left(1/2+b_x\gamma\right)\sigma^2\phi'(x) + \sigma z$ \\
    \qquad Output $y$
\end{tabbing}
\end{algorithm}

The procedure can be viewed as taking elements of both the Langevin and Barker proposal, as the standard Langevin proposal is modified according to a random variable that takes value $+1$ or $-1$ with some probability that is skewed in the direction of the gradient as in the Barker proposal.
An interesting aspect of the Gaussian noise $\mu$ is that it is independent of the coordinate system, making the Algorithm~\ref{alg:g_family} depend on the coordinate system only through the choice of the flipping directions.

\begin{proposition} \label{prop:g_family_result}
Algorithm \ref{alg:g_family} produces a sample from a distribution with density proportional to $\sigma^{-1}\mu((y-x)/\sigma)g_\gamma(e^{\phi'(x)(y-x)})$. %The corresponding normalising constant is $\frac{1}{2}\left(e^{\frac{1}{2}\sigma^2(\frac{1}{2}-\gamma)^2\phi'(x)^2}+e^{\frac{1}{2}\sigma^2(\frac{1}{2}+\gamma)^2\phi'(x)^2}\right)$.
\end{proposition}
\begin{proof}
First note that
\[
p(x)
%=
%\frac{e^{-\frac12\gamma\sigma^2\phi'(x)}}{e^{-\frac12\gamma\sigma^2\phi'(x)}+e^{\frac12\gamma\sigma^2\phi'(x)}}
=
\frac{e^{-\frac12(\gamma+\frac12)^2\sigma^2\phi'(x)}}{e^{-\frac12(\gamma+\frac12)^2\sigma^2\phi'(x)}+e^{\frac12(\gamma-\frac12)^2\sigma^2\phi'(x)}}
\]
so that
\[
p(x)e^{-\frac12(\gamma+\frac12)^2\sigma^2\phi'(x)}
=
(1-p(x))e^{-\frac12(\gamma-\frac12)^2\sigma^2\phi'(x)}\,.
\]
The random variable $Y$ is a mixture of two Gaussians with the same variance. Setting $\mu_x^+ = x+(1/2 + \gamma)\sigma^2\phi'(x)$ and $\mu_x^- = x+(1/2 - \gamma)\sigma^2\phi'(x)$, the density of $\xi(y)$ of $Y$ satisfies
\begin{align*}
\sigma\sqrt{2\pi}\xi(y)
&= 
p(x)e^{-\frac{1}{2\sigma^2}(y-\mu_x^+)^2}
+
\left(1-p(x)\right)e^{-\frac{1}{2\sigma^2}(y-\mu_x^-)^2}
\\
&=
e^{-\frac{(y-x)^2}{2\sigma^2}}\left(p(x)e^{-\frac{1}{2}\left(\gamma+\frac{1}{2}\right)^2\sigma^2\phi'(x)^2}\right)\left(e^{\left(\frac{1}{2}+\gamma\right)(y-x)\phi'(x)}+e^{\left(\frac{1}{2}-\gamma\right)(y-x)\phi'(x)}\right)
\\
&\propto
%\frac{2}{e^{\frac{1}{2}\sigma^2(\frac{1}{2}-\gamma)^2\phi'(x)^2}+e^{\frac{1}{2}\sigma^2(\frac{1}{2}+\gamma)^2\phi'(x)^2}}\times 
\frac{1}{2}\left((e^{(y-x)\phi'(x)})^{\frac{1}{2}+\gamma}+(e^{(y-x)\phi'(x)})^{\frac{1}{2}-\gamma}\right)
\times e^{-\frac{(y-x)^2}{2\sigma^2}}
\\
&\propto 
%\underbrace{\frac{2}{e^{\frac{1}{2}\sigma^2(\frac{1}{2}-\gamma)^2\phi'(x)^2}+e^{\frac{1}{2}\sigma^2(\frac{1}{2}+\gamma)^2\phi'(x)^2}}}_{Z_{\sigma}(x)^{-1}}\times
g_\gamma(e^{\phi'(x)(y-x)})\sigma^{-1}\mu((y-x)/\sigma)
\end{align*}
as required.
\end{proof}

\section{Illustration of different choices of $\mu_4$ in the optimal choice of first order locally-balanced proposal}
Here we compare the three point proposal with different choices of $\mu_4$ on the Hyperbolic target distribution considered in Section 5.1 of the paper. 
For any value of $\mu_4$, the optimal value of $\mathfrak{g}$ as given by Proposition 5 in the paper is used.
%Below is a simple example of the expected squared jump distance for a finite dimensional Hyperbolic target distribution, using the three point proposal with different choices of $\mu_4$. 
Theory suggests that optimal asymptotic performances are given by choosing $\mu_4$ arbitrarily close to $1$ but, as discussed in the paper, values of  $\mu_4$ close to $1$ may need very large dimensionality for the actual asymptotic regime to kick in. 
\begin{figure}[h!]
\centering
\includegraphics[width=0.9\linewidth]{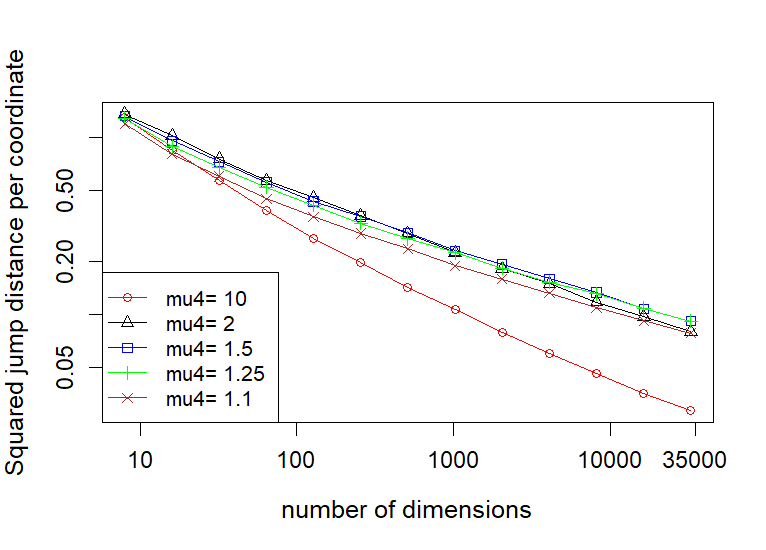}
\caption{Expected squared jump distance against dimensionality for various three point proposals on an hyperbolic product target distribution.
}
\label{fig:esjd_mu4_hyperb}
\end{figure}
Figure \ref{fig:esjd_mu4_hyperb} illustrates this phenomenon. 
For small dimensionality, the value $\mu_4=2$ is best (among the ones considered) but as dimensionality increases $\mu_4=1.5$ starts outperforming and then $\mu_4=1.25$. Dimension $n=35000$ is not yet sufficient to for $\mu_4=1.1$ to outperform $\mu_4=1.25$, although theory suggests that this will eventually happen.

%Figure \ref{fig:esjd_mu4_hyperb_zoom} illustrates the expected squared jump distance times $n^{1/3}$ [[DO WE WANT TO PUT ALSO THIS? IT MAGNIFIES THE EFFECT, BUT THE PUZZLING THING IS THAT IT SEEMS TO GROW ALSO FOR HYPERB, NOT ONLY FOR GAUSS (IT SHOULDN'T, NO?)]
%\begin{figure}[h!]
%\centering
%\includegraphics[width=0.49\linewidth]{Scaling_various_mu4_hyperb_zoom.png}
%\includegraphics[width=0.49\linewidth]{Scaling_various_mu4_Gauss_zoom.png}
%\caption{Left: same as Figure \ref{fig:esjd_mu4_hyperb} but with expected square jump distance times $n^{1/3}$ on the $y$-axes.
%Right: same for Gaussian target.
%}
%\label{fig:esjd_mu4_hyperb_zoom}
%\end{figure}

%WE HAVE ALSO THE SAME PLOTS FOR GAUSSIAN TARGETS, BUT THERE IT IS LESS CLEAR THE THEORETICAL STORY BECAUSE ALL VALUES OF MU4 LEAD TO THETA=0, SO WE CANNOT SAY MU4 CLOSE TO 1 WILL BE EVENTUALLY BETTER (INDEED IT SEEMS TO BE NOT THE CASE)

%\eqref{eq:1storderlb}

%\appendixone
\section{Miscellaneous proofs}

\subsection{Section \ref{sec:lb}}

\begin{proof}[Proof of Proposition \ref{prop:g_to_even}]
For the first part of the statement, direct calculation gives
$$
tg_h(1/t) = t \times t^{-1/2} h(-\log t) = t^{1/2} h(\log t) = g_h(t).
$$
Similarly for the second part
$$
h_g(-x) = e^{x/2}g_h(e^{-x}) = e^{-x/2}g_h(e^x) = h_g(x)
$$
using that $g_h(e^{-x}) = e^{-x}g_h(e^x)$. Noting also that $h_g:\R \to [0,\infty)$ shows that $h_g \in \mathcal{H}$.
\end{proof}

\subsection{Section \ref{sec:universal}}

\begin{proof}[Proof of Theorem \ref{thm:log-MH}]
See Section~3, particularly the proof of Theorem~8, in \cite{vogrinc2021counterexamples}.
\end{proof}

\begin{proof}[Proof of Proposition \ref{prop:ass1sufficient}]
(i) Easy to verify.

(ii) Since $\int_{\R}z^\xi \mu(dz)<\infty$, we have
\[
\int_{\R}e^{b(az)}|z|^\xi \mu(dz)= \int_{\R}g(e^{az})|z|^\xi \mu(dz)
\leq
\|g\|_\infty\int_{\R}|z|^\xi \mu(dz)\,.
\]
On the other hand $b$ is non decreasing and $b(0)=0$ so
\[
\int_{\R}e^{b(az)}\mu(dz)
\geq
\int_0^\infty \mu(dz)
=
\frac{1}{2}\,.
\]

(iii) This essentially holds because $b'$ is bounded and hence $b$ is at most linear. Clearly $\mu$ has finite polynomial moments.
Next note that by fundamental theorem of calculus and $b(0)=0$ we have
\[
b(az)
=
\int_0^1 azb'(azs)ds\,.
\]
Hence
\begin{align*}
    \int_{\R}e^{b(az)}|z|^\xi \mu(z)dz
    &=
    \int_{\R}e^{az\int_0^1b'(saz)ds}|z|^\xi \mu(z)dz
    \\
    &\leq
    \tilde C_\mu\left(1+\left|a\int_0^1b'(saz)ds\right|^{\tilde\beta}\right)\int_{\R}e^{az\int_0^1b'(saz)ds}\mu(z)dz
    \\
    &\leq
    \tilde C_\mu\max(1,\|b'\|_\infty^{\tilde{\beta}})\left(1+|a|^{\tilde\beta}\right)\int_{\R}e^{b(az)}\mu(z)dz\,.
\end{align*}

(iv) We will prove this implies a special instance of (iii). Due to symmetry of $\mu$ we may assume without loss of generality that $a>0$.
Firstly, $\mu$ has finite moment generating function since for any $b>a$ and $c>0$ we have
\[
\int_\R e^{az}\mu(z)dz
\leq
2\int_0^\infty e^{bz}e^{(a-b)z}\mu(z)dz
\leq
2c \left(\sup_{z\leq c}e^{az}\mu(z)\right) + \frac{2}{b-a}e^{-c(b-a)} \left(\sup_{z> c}e^{bz}\mu(z)\right)\,.
\]

Secondly, we will use inequalities $|z|^{\lambda}1_{[-1,1]^c}(z)\leq |z|^{\lambda+1}$ and  $e^a=e^{a\int_0^\infty2\mu(z)dz}\leq 2\int_0^\infty e^{az}\mu(z)dz\leq  2\int_\R e^{az}\mu(z)dz$, which follows by Jensen's inequality. Integration by parts implies that for any $\lambda>p-1$ we have
\begin{align*}
    \int_\R e^{az} |z|^{\lambda} \mu(z)dz
    &\leq
    \int_{[-1,1]} e^{az}\mu(z)dz+
    A\int_{[-1,1]^c} e^{az}|z|^{\lambda-p}\mu(z)dz-B\int_{[-1,1]^c} e^{az}|z|^{\lambda-p}z\mu'(z)dz
    \\
    &=
    \int_{[-1,1]} e^{az}\mu(z)dz+
    A\int_{[-1,1]^c} e^{az}|z|^{\lambda-p}\mu(z)dz
    \\
    &\qquad +
    B\int_{[-1,1]^c} e^{az}((\lambda-p+1)|z|^{\lambda-p}+a|z|^{\lambda-p}z)\mu(z)dz
    +B\mu(1)(e^a-e^{-a})
    \\
    &\leq
    \left(1+2B\mu(1)\right)\int_{\R} e^{az}\mu(z)dz+
   \left(A+B(\lambda-p+1)+ |a|B\right)\int_{\R} e^{az}|z|^{\lambda-p+1}\mu(z)dz\,.
\end{align*}
Thus, we have successfully reduced the power of $z$ in the integrand by $p-1$ at the expense of producing an $|a|$. Noting that $|z|^\xi\leq 1+|z|^{\lceil\frac{\xi}{p-1}\rceil (p-1)}$ and recursively applying the same argument establishes the claim. 
\end{proof}

\begin{proof}[Proof of Theorem \ref{thm:CLT}]

Denote $b(x)=\log g(e^x)$ and note that $b(0)=0$, $b'(0)=1/2$ and $b''(0)=1/4 +g''(1)$. First consider $n$ as fixed. By definition for any $y\in\R$
\begin{equation}\label{eq:rho_general}
    \rho_n(x,y)
=
\phi(y)-\phi(x)+b(\phi'(y)(x-y))-b(\phi'(x)(y-x))-\log(Z_{\sigma_n}(x))+\log(Z_{\sigma_n}(y))\,.
\end{equation}
We will be using the fundamental theorem of calculus
\begin{equation}\label{eq:ExactTaylor}
    U(1)-U(0)
    =
    \int_0^1U'(s)ds
\end{equation}
for various functions $U$. 
Denote $w=(y-x)/\sigma_n$, and note that $\Zs(x)=\int_\R e^{b(\phi'(x)\sigma_nz)}\mu(dz)$ by an analogous substitution. Then note that
$$
    \log\left(\frac{\Zs(x+\sigma_n w)}{\Zs(x)}\right)
    =
    \sigma^2_n w\int_{0}^1\phi''(x+t\sigma_n w)\frac{\int_\R e^{b(\phi'(x+t\sigma_n w)\sigma_n z)}b'(\phi'(x+t\sigma_n w)\sigma_n z)z\mu(dz)}{\int_\R e^{b(\phi'(x+t\sigma_n w)\sigma_n z)}\mu(dz)}dt
$$
using \eqref{eq:ExactTaylor} with the function $U_1(t)=\log(\Zs(x+t\sigma_n w)$. Using \eqref{eq:ExactTaylor} again with the function $U_2(s)=b'(s\phi'(x+t\sigma_n w)\sigma_n z)$ and recalling that $b'(0)= 1/2$, this can be written as
\begin{align*}
    &\frac{\sigma_n^2}{2}w\int_{0}^1\phi''(x+t\sigma w)\frac{\int_\R e^{b(\phi'(x+t\sigma_n w)\sigma_n z)}z\mu(dz)}{\int_\R e^{b(\phi'(x+t\sigma_n w)\sigma_n z)}\mu(dz)}dt
    \\
    &\quad +
    \sigma_n^3w\iint_{[0,1]^2}\phi''(x+t\sigma_n w)\phi'(x+t\sigma_n w)\frac{\int_\R e^{b(\phi'(x+t\sigma_n w)\sigma_n z)} b''(s\phi'(x+t\sigma_n w)\sigma_n z) z^2\mu(dz)}{\int_\R e^{b(\phi'(x+t\sigma_n w)\sigma_n z)}\mu(dz)}dtds.
\end{align*}
Using \eqref{eq:ExactTaylor} in the numerator of the first term with the function $U_3(u)=e^{b(u\phi'(x+t\sigma_n w)\sigma_n z)}$ and recalling that $b(0)=0$ gives
\begin{align*}
&\frac{\sigma_n^3w}{2}\iint_{[0,1]^2}\phi''(x+t\sigma_n w)\phi'(x+t\sigma_n w)\frac{\int_\R e^{b(u\phi'(x+t\sigma_n w)\sigma_n z)}b'(u\phi'(x+t\sigma_n w)\sigma_n z) z^2\mu(dz)}{\int_\R e^{b(\phi'(x+t\sigma_n w)\sigma_n z)}\mu(dz)}dtdu
\\
&\qquad +
\sigma_n^3w\iint_{[0,1]^2}\phi''(x+t\sigma_n w)\phi'(x+t\sigma_n w)\frac{\int_\R e^{b(\phi'(x+t\sigma_n w)\sigma_n z)} b''(s\phi'(x+t\sigma_n w)\sigma_n z) z^2\mu(dz)}{\int_\R e^{b(\phi'(x+t\sigma_n w)\sigma_n z)}\mu(dz)}dtds.
\end{align*}
And further simplification leads to the expression
\begin{align*}
\sigma_n^3w & \iint_{[0,1]^2}\phi''(x+t\sigma_n w)\phi'(x+t\sigma_n w)
\\
&\qquad \times
    \frac{\int_\R e^{b(\phi'(x+t\sigma_n w)\sigma_n z)} \left(\frac{1}{2}b'(s\phi'(x+t\sigma_n w)\sigma_n z)+b''(s\phi'(x+t\sigma_n w)\sigma_n z)\right) z^2\mu(dz)}{\int_\R e^{b(\phi'(x+t\sigma_n w)\sigma_n z)}\mu(dz)}dtds 
\\
&= 
    \sigma_n^3w\phi''(x)\phi'(x)\left(\frac{1}{2}b'(0)+b''(0)\right)
    + R_{1,n}(x,w) + R_{2,n}(x,w) + R_{3,n}(x,w)
\\
&= 
    \sigma_n^3w\phi''(x)\phi'(x)\left(\frac{1}{2}+g''(1)\right)
    + R_{1,n}(x,w) + R_{2,n}(x,w) + R_{3,n}(x,w),
\end{align*}
where
\begin{align*}
    R_{1,n}(x,w)
    &=
    \sigma_n^3w\left(\frac{1}{2}b'(0)+b''(0)\right)\int_0^1\left(\phi''(x+t\sigma_n w)\phi'(x+t\sigma_n w)-\phi''(x)\phi'(x)\right)dt
    \\
    R_{2,n}(x,w)
    &=
    \sigma_n^3w\left(\frac{1}{2}b'(0)+b''(0)\right)\int_0^1\phi''(x+t\sigma_n w)\phi'(x+t\sigma_n w)
    \frac{\int_\R e^{b(\phi'(x+t\sigma_n w)\sigma_n z)}  \left(z^2-1\right)\mu(dz)}{\int_\R e^{b(\phi'(x+t\sigma_n w)\sigma_n z)}\mu(dz)}dt
    \\
    R_{3,n}(x,w)
    &=
    \sigma_n^3w\int_{[0,1]^2}\phi''(x+t\sigma_n w)\phi'(x+t\sigma_n w)dt
    \\
    \times &
    \frac{\int_\R e^{b(\phi'(x+t\sigma_n w)\sigma_n z)}  \left(\frac{1}{2}b'(s\phi'(x+t\sigma_n w)\sigma_n z)+b''(s\phi'(x+t\sigma_n w)\sigma_n z)-\frac{1}{2}b'(0)-b''(0)\right) z^2\mu(dz)}{\int_\R e^{b(\phi'(x+t\sigma_n w)\sigma_n z)}\mu(dz)}dtds\,.
\end{align*}
%For equality (a) we used \eqref{eq:ExactTaylor} with the function $U_1(t)=\log(\Zs(x+t\sigma_n w)$. For (b) we used \eqref{eq:ExactTaylor} with the function $U_2(s)=b'(s\phi'(x+t\sigma_n w)\sigma_n z)$ with $b'(0)= 1/2$. For (c) we used \eqref{eq:ExactTaylor} in the numerator in the first term with function $U_3(u)=e^{b(u\phi'(x+t\sigma_n w)\sigma_n z)}$ with $b(0)=0$.

To expand the remaining terms of \eqref{eq:rho_general} we use the exact third order Taylor expansion
\begin{align*}
b(t)
&= 
b(0)+tb'(0)+t^2\frac{b''(0)}{2}+t^3\frac{b'''(0)}{6}+\frac{t^3}{2}\int_{0}^1(b'''(ut)-b'''(0))(1-u)^2du 
\\
&= 
b(0)+tb'(0)+t^2\frac{b''(0)}{2}+\frac{t^3}{2}\int_{0}^1b'''(ut)(1-u)^2du
\end{align*}
on each of the terms involving $b$ to obtain
\begin{align*}
    %\phi(x+\sigma_n w)-\phi(x)
    %+
    b(-\phi'(x+\sigma_n w)\sigma_n w)-b(\phi'(x)\sigma_n w)
    &= 
    %\phi(x+\sigma_n w)-\phi(x)
    -\frac{\sigma_n w}{2}\left(\phi'(x)+\phi'(x+\sigma_n w)\right)
    \\
    &\qquad + 
    \sigma_n^2 w^2\left(\frac{1}{8}+\frac{g''(1)}{2}\right)\left(\phi'(x+\sigma_n w)^2-\phi'(x)^2\right)
    \\
    &\qquad +  R_{6,n}(x,w)+R_{7,n}(x,w)\,,
    \end{align*}
where 
\begin{align*}
    R_{6,n}(x,w) &= 
    -\frac{1}{2}\sigma_n^3 w^3\phi'(x+\sigma_n w)^3\int_0^1b'''(-u\phi'(x+\sigma_n w)\sigma_n w)(1-u)^2du\,,
    \\
    R_{7,n}(x,w) &= 
    -\frac{1}{2}\sigma_n^3w^3\phi'(x)^3\int_0^1b'''(u\phi'(x)\sigma_n w)(1-u)^2du\,.
\end{align*}
Taking one half of the difference of the exact Taylor expansions
\begin{align*}
\phi(x+\sigma_n w)-\phi(x)	 
&= 		
\sigma_n w\phi'(x) + \sigma_n^2w^2\int_0^1\phi''(x+u\sigma_n w)(1-u)du\,,
\\
\phi(x)-\phi(x+\sigma_n w)	 
&=
-\sigma_n w\phi'(x+\sigma_n w) + \sigma_n^2w^2\int_0^1\phi''(x+(1-u)\sigma_n w)(1-u)du
\end{align*}
and setting $v=1-u$ reveals that
\begin{align*}
\phi(x+\sigma_n w)-\phi(x)-\frac{\sigma_n w}{2}\left(\phi'(x)+\phi'(x+\sigma_n w)\right)
&= 
\frac{\sigma_n^2w^2}{2}\int_0^1\phi''(x+u\sigma_n w)(1-2u)du
\\
&= 
\frac{\sigma_n^2w^2}{2}\int_0^1\left(\phi''(x+u\sigma_n w)-\phi''(x)\right)(1-2u)du
\\
&= 
\frac{\sigma_n^3w^3}{2}\int_0^1u(1-2u)\int_0^1\phi'''(x+uv\sigma_n w)dvdu
\\
&= 
-\frac{\sigma_n^3w^3}{12}\phi'''(x)
 +  R_{4,n}(x,w),
\end{align*}
where
\[
R_{4,n}(x,w)
 = 
\frac{\sigma_n^3w^3}{2}\int_0^1u(1-2u)\int_0^1\left(\phi'''(x+uv\sigma_n w)-\phi'''(x)\right)dvdu\,.
\]
Similarly,
\begin{align*}
\sigma_n^2 w^2\left(\frac{1}{8}+\frac{g''(1)}{2}\right)\left(\phi'(x+\sigma_n w)^2-\phi'(x)^2\right)
&= 
\sigma_n^3 w^3\left(\frac{1}{4}+g''(1)\right)\int_0^1 \phi'\phi''(x+u\sigma_n w)du
\\
&=  
\sigma_n^3 w^3\left(\frac{1}{4}+g''(1)\right)\phi'(x)\phi''(x)
 + 
R_{5,n}(x,w),
\end{align*}
where
\[
R_{5,n}(x,w)
 = 
\sigma_n^3 w^3\left(\frac{1}{4}+g''(1)\right)\int_0^1 \left(\phi'\phi''(x+u\sigma_n w)-\phi'\phi''(x)\right)du.
\]

 Next, we will denote with $\tilde\E_n$ the expectation with respect to the measure $\Zs^{-1}(x)e^{b(\phi'(x)\sigma_n w)}\pi(x)dx\mu(dw)$. We will show that $\tilde\E_n[R_{i,n}^2]$ decays faster than $\sigma_n^6$ for $i=1,2,\dots,7$ as $n\to\infty$. Note that Assumption~\ref{ass1} iii) is the tool enabling the control of the Radon-Nikodym derivatives of various proposals (indexed by $n$) with respect to $\mu$, making it possible for us to represent these expectations with respect to the measure $\pi(x)dx\mu(dw)$. 
Using $a \lesssim b$ to denote that $a \leq cb$ for some positive finite $c$ independent of $n$, the bound for the first term $R_{1,n}$ follows by Assumption~\ref{ass1} i) and iii) as
\begin{align*}
    \tilde\E_n[R_{1,n}^2]
    &\lesssim 
    \sigma_n^6
    \iint_{\R^2}\int_0^1\left(\phi''(x+t\sigma_n w)\phi'(x+t\sigma_n w)-\phi''(x)\phi'(x)\right)^2w^2\frac{e^{b(\phi'(x)\sigma_n w)}}{\Zs(x)}dt\mu(dw)\pi(x)dx
    \\
    &\leq 
    \sigma_n^6
    \iint_{\R^2}\int_0^1K(x)^2w^2\max(|t\sigma_n w|^H,|t\sigma_n w|^\gamma)^2\frac{e^{b(\phi'(x)\sigma_n w)}}{\Zs(x)}dt\mu(dw)\pi(x)dx
    \\
    &\lesssim 
    \sigma_n^{6+2H}
    \int_\R K(x)^2\pi(x) \int_\R |w|^\xi\frac{e^{b(\phi'(x)\sigma_n w)}}{\Zs(x)}\mu(dw)dx
    \\
    &\lesssim
    \sigma_n^{6+2H+\beta}
    \int_\R K(x)^2(1+|\phi'(x)|^\beta)\pi(x)dx\,.
\end{align*}
To bound the second term $\tilde\E_n R_{2,n}$, note that $\lim_{\sigma_n\to 0} \sigma_n^{-3}R_{2,n}=0$ for every $x$ and $w$, so we only need to provide a dominating bound for the integrand.  Using Assumption~\ref{ass1} iii) gives that this is
\begin{multline*}
    %\iint_{\R^2}\left(\int_0^1\phi''(x+t\sigma_n w)\phi'(x+t\sigma_n w)\frac{\int_\R e^{b(\phi'(x+t\sigma_n w)\sigma_n z)}  \left(z^2-1\right)\mu(dz)}{\int_\R e^{b(\phi'(x+t\sigma_n w)\sigma_n z)}\mu(dz)}dt\right)^2 w^2\frac{e^{b(\phi'(x)\sigma_n w)}}{\Zs(x)}\mu(dw)\pi(x)dx
    \lesssim 
    \iint_{\R^2}\int_0^1\phi''(x+t\sigma_n w)^2\phi'(x+t\sigma_n w)^2\left(1+\sigma_n^{2\beta}\phi'(x+t\sigma_n w)^{2\beta}\right) \frac{e^{b(\phi'(x)\sigma_n w)}}{\Zs(x)}w^2dt\mu(dw)\pi(x)dx.
\end{multline*}
To show that this is finite we use Assumption~\ref{ass1}~i) to control the perturbation of functions $\phi''\phi'$ and $\phi''|\phi'|^{1+\beta}$ from $x$ to $x+t\sigma_n w$ and  Assumption~\ref{ass1}~iii) to control the Radon-Nikodym derivative $\Zs(x)^{-1}e^{b(\phi'(x)\sigma_nw)}$ (as in the argument for the term $\tilde\E_n[R_{1,n}]$):
\begin{align*}
    &\iint_{\R^2}\int_0^1\phi''(x+t\sigma_n w)^2\phi'(x+t\sigma_n w)^2\left(1+|\phi'(x+t\sigma_n w)|^{\beta}\right)^2 w^2\frac{e^{b(\phi'(x)\sigma_n w)}}{\Zs(x)}dt\mu(dw)\pi(x)dx
    \\
    &\quad\lesssim 
    \iint_{\R^2}\phi''(x)^2\phi'(x)^2\left(1+|\phi'(x)|^{\beta}\right)^2 w^2\frac{e^{b(\phi'(x)\sigma_n w)}}{\Zs(x)}\mu(dw)\pi(x)dx
    \\
    &\qquad +
    \iint_{\R^2}\int_0^1\left(\phi''(x+t\sigma_n w)\phi'(x+t\sigma_n w)-\phi''(x)\phi'(x)\right)^2dtw^2\frac{e^{b(\phi'(x)\sigma_n w)}}{\Zs(x)}\mu(dw)\pi(x)dx
    \\
    &\qquad +
    \iint_{\R^2}\int_0^1\left(\phi''(x+t\sigma_n w)|\phi'(x+t\sigma_n w)|^{1+\beta}-\phi''(x)|\phi'(x)|^{1+\beta}\right)^2dtw^2\frac{e^{b(\phi'(x)\sigma_n w)}}{\Zs(x)}\mu(dw)\pi(x)dx
    \\
    &\quad \lesssim 
    \int_\R \phi''(x)^2\phi'(x)^2(1+|\phi'(x)|^{\beta})^3\pi(x)dx+\sigma_n^{2H+\beta}
    \int_\R K(x)^2\pi(x)\int_\R |w|^\xi\frac{e^{b(\phi'(x)\sigma_n w)}}{\Zs(x)}\mu(dw)dx
    \\
    &\quad \lesssim
    \int_\R \phi''(x)^2\phi'(x)^2(1+|\phi'(x)|^{\beta})^3\pi(x)dx+\int_\R K(x)^2(1+|\phi'(x)|^\beta)\pi(x)dx.
\end{align*}
The term $\tilde\E[R_{3,n}]$ is handled similarly with the 
dominated convergence theorem as $\sigma_n^{-3}R_3$ converges point-wise to zero as $\sigma_n\to0$. The dominating bound is very similar as for $R_{2,n}$ once we notice that $b'$ and $b''$ are both bounded by Assumption~\ref{ass1} ii).

By Assumption~\ref{ass1}~i) the terms $R_{4,n}$ and $R_{5,n}$ are bounded absolutely by a constant multiplier of \\ $K(x) \sigma_n^{3+H}\max(|w|^{3+H},|w|^{3+\gamma})$, hence $\lim_{\sigma_n\to 0}\sigma_n^{-6}\tilde\E_n[R_{4,n}^2]=\lim_{\sigma_n\to 0}\sigma_n^{-6}\tilde\E_n[R_{5,n}^2]=0$\,.
The square of $R_{7,n}$ can be seen to decay faster than $\sigma_n^6$ by the Dominated convergence theorem, since $b'''$ is bounded, $b'''(0)=0$ and integrability guaranteed by Assumption~\ref{ass1}~i). Similarly
\begin{multline*}
R_{6,n} = -\frac{1}{2}\sigma_n^3 w^3\phi'(x)^3\int_0^1b'''(-u\phi'(x+\sigma_n w)\sigma_n w)(1-u)^2du\\
-
\frac{1}{2}\sigma_n^3 w^3\left(\phi'(x+\sigma_n w)^3-\phi'(x)^3\right)\int_0^1b'''(-u\phi'(x+\sigma_n w)\sigma_n w)(1-u)^2du\,,
\end{multline*}
so again the square of the first part decays faster than $\sigma_n^6$ by the Dominated convergence theorem (as for $R_{7,n}$) and the second part is dominated by a constant multiplier of $K(x)\sigma_n^{3+H}\max(|w|^{3+H},|w|^{3+\gamma)}$ due to Assumption~\ref{ass1} i).

Write
\[
T(x,w)=w^3\left(-\frac{1}{12}\phi'''(x)+\left(\frac{1}{4}+g''(1)\right)\phi'(x)\phi''(x)\right)
-
w\left(\frac12+g''(1)\right)\phi'(x)\phi''(x)
\]
and set $R_n(x,w)=\sum_{i=1}^7R_{i,n}(x,w)$. Clearly
\[
\tilde\E_n[\rho_n^2]
 = 
\tilde\E_n\left[(\sigma^3T+R_n)^2\right]
 = 
\tilde{\E}_n\left[\sigma^6T^2\right]
+2\tilde{\E}_n\left[\sigma^3TR_n\right]
+\tilde{\E}_n\left[R_n^2\right].
\]
Using the inequalities $\tilde{\E}_n \left[ TR_n \right]^2 \leq \tilde{\E}_n[T^2]\tilde{\E}_n[R_n^2]$
and
$\tilde{\E}_n[R_n^2]
 \leq 
7\sum_{i=1}^7\tilde{\E}_n[R_{i,n}^2]$
we see that the last two terms decay faster than $\sigma_n^{-6}$. We can now identify the limiting $\theta^2$ using the Dominated convergence theorem. We show that $\lim_{n\to\infty}\sigma_n^{-6}\E[\rho_n^2]=\lim_{n\to\infty}\tilde{\E}_n[T^2]=\E[T^2]=\theta^2$ where the last expectation is with respect to $\pi(x)dx\mu(dw)$. Indeed,
\[
    \tilde{\E}_n[T^2]
    =
    \int_{\R}T^2(x,w)\frac{e^{\sigma_n w\phi'(x)}}{\Zs(x)}\mu(dw)\pi(x)dx\,.
\]
The integrand converges point-wise to $T^2(x,w)$ and is dominated by $T^2(x,w)(1+|\phi'(x)|^\beta)$ which is integrable by Assumption\ref{ass1} i). Expanding the expression for $T$ and integrating with respect to $\mu(dw)\pi(x)dx$ leads to the claimed form of $\theta^2$.

Finally, to finish the proof using Theorem~\ref{thm:log-MH} notice that by choosing $\sigma_n$ such that  $n\sigma_n^6\to\ell^6$ we get
\[
\lim_{n\to\infty}n\E[\rho_n^2]
=
\lim_{n\to\infty}n\sigma_n^6\sigma_n^{-6}\E[\rho_n^2]
=
\ell^6\theta^2\,.
\]
We also need to prove that $\lim_{n\to\infty}\sigma_n^{-6}\tilde{\E}_n\left[\rho_n^21_{\rho_n<-\sigma_n}\right]=0$. Indeed,
\begin{align*}
\tilde{\E}_n\left[\rho_n^21_{\rho_n<-\sigma_n}\right]
&= 
\sigma_n^6\tilde{\E}_n\left[T^21_{\rho_n<-\sigma_n}\right]+2\sigma_n^3\tilde{\E}_n\left[TR_n1_{\rho_n<-\sigma_n}\right]+\tilde{\E}_n\left[R_n^21_{\rho_n<-\sigma_n}\right]
\\
&\leq 
\sigma_n^6\tilde{\E}_n\left[T^21_{\rho_n<-\sigma_n}\right]+2\sigma_n^3\sqrt{\tilde{\E}_n[T^2]\tilde{\E}_n[R_n^2]}+2\E[R_n^2].
\end{align*}
We have already established that the last two terms decay faster than $\sigma_n^{-6}$. The first term can be bounded using the H\"older and Markov inequalities by
\begin{align*}
\sigma_n^6\tilde{\E}_n\left[T^21_{\rho_n<-\sigma_n}\right]
&\leq 
\sigma_n^6\tilde{\E}_n[T^{2+\epsilon}]^{\frac{2}{2+\epsilon}}\PP[\rho_n<-\sigma_n]^{\frac{\epsilon}{2+\epsilon}}
\\
&\leq 
\sigma_n^6\tilde{\E}_n[T^{2+\epsilon}]^{\frac{2}{2+\epsilon}}\left(\frac{1}{\sigma_n^2}\tilde{\E}_n[\rho_n^2]\right)^{\frac{\epsilon}{2+\epsilon}}
\\
&= 
\sigma_n^{6+\frac{4\epsilon}{2+\epsilon}}\tilde{\E}_n[T^{2+\epsilon}]^{\frac{2}{2+\epsilon}}\left(\frac{1}{\sigma_n^6}\tilde\E_n[\rho_n^2]\right)^{\frac{\epsilon}{2+\epsilon}}.
\end{align*}
Analogously as before $\tilde{\E}_n[T^{2+\epsilon}]\to\E[T^{2+\epsilon}]<\infty$ by the Dominated convergence theorem and Assumption~\ref{ass1} i) and we have already established that $\sigma_n^{-6}\tilde{\E}_n[\rho_n^2]\to \theta^2$. Together these establish that
$$
\lim_{n\to\infty}\sigma_n^{-6}\tilde{\E}_n\left[\rho_n^21_{\rho_n<-\sigma_n}\right]=0
$$
which completes the proof.

%(iv) The Central limit theorem follows from (iii) and Theorem~\ref{thm:log-MH}. The average acceptance rate then follows as a consequence of weak convergence and Proposition~2.4 in \cite{roberts1997weak}:
%\begin{multline*}
%\E[\alpha_n]
% = 
%\E\left[(1\wedge e)\left(\sum_{i=1}^n\rho_{n}(X_i,Y_i^{(n)})\right)\right]
%\\\xrightarrow{n\to\infty} 
%\E\left[(1\wedge e)\left(N\left(-\frac{1}{2}\ell^6\theta^2,\ell^6\theta^2\right)\right)\right]
% = 
%2\Phi\left(-\frac{\ell^3\theta}{2}\right)\,.
%\end{multline*}
\end{proof}

\begin{proof}[Proof of Theorem \ref{thm:ESJD}]
We fix $g$ and $\mu$ that satisfy Assumption~\ref{ass1} and suppress the notation with respect to them for the rest of the proof. 
First assume $\lim_{n\to\infty}n^{1/6}\sigma_n=0$, which implies that
\begin{equation}\label{eq:aux1}
    \sigma_n^{-2}\E\left[(Y_{n,1}-X_{n,1})^2\right]= 
    \sigma_n^{-2}\int_{\R}\int_{\R}Z_{\sigma_n}^{-1}(x)g(e^{\phi'(x)w\sigma_n})w^2\mu\left(dw\right)\pi(x)dx
     \to
    1
\end{equation}
as $n\to \infty$.  Convergence to one follows by the Dominated convergence theorem as $\sigma_n\to 0$. The dominating bound is provided by Assumption~\ref{ass1} iii). Hence, $\E\left[(Y_{n,1}-X_{n,1})^2\right]$ decays as $\sigma_n^2$ which is by definition faster than $n^{-1/3}$. Because the acceptance rate is bounded above by one, we must have that $\lim_{n\to\infty}n^{1/3}\mathcal{E}_n=0$.

Next assume $\lim_{n\to\infty}n^{1/6}\sigma_n=\infty$. First we will establish that
\begin{equation}\label{eq:aux2}
    \E\left[(Y_{n,1}-X_{n,1})^2\left(1\wedge e\right)\left(\sum_{i=1}^n\rho_{n,i}\right)\right]
    \leq 
    2\E\left[(Y_{n,1}-X_{n,1})^2\right]\E\left[\left(1\wedge e\right)\left(\sum_{i=2}^n\rho_{n,i}\right)\right]\,.
\end{equation}
The statement (without the two) is clear if $\rho_{n,1}$ is negative. When it is positive then using Lemma~\ref{lem:de-exponentialisation} with respect to the log-Metropolis--Hastings random variable $\rho_{n,1}$ and $h(x,y)=(y-x)^2$ gives
\begin{align*}
    \E\left[(Y_{n,1}-X_{n,1})^2\left(1\wedge e\right)\left(\sum_{i=1}^n\rho_{n,i}\right)1_{\rho_{n,1}\geq 0}\right]
    &\leq 
    \E\left[(Y_{n,1}-X_{n,1})^2e^{\rho_{n,1}}
    \cdot\left(1\wedge e\right)\left(\sum_{i=2}^n\rho_{n,i}\right)1_{[0,\infty)}\left(\rho_{n,1})\right)\right]
    \\
    &= 
    \E\left[(Y_{n,1}-X_{n,1})^2\left(1\wedge e\right)\left(\sum_{i=2}^n\rho_{n,i}\right)1_{[0,\infty)}\left(-\rho_{n,1}\right)\right]
    \\
    &\leq 
    \E\left[(Y_{n,1}-X_{n,1})^2\right]\E\left[\left(1\wedge e\right)\left(\sum_{i=2}^n\rho_{n,i}\right)\right]\,.
\end{align*}
Adding terms corresponding to the sign of $\rho_{n,1}$ implies \eqref{eq:aux2}.

We still need to bound the acceptance rate. To do this we split the probability space regarding the event $\mathcal{A}_n:=\{ \sum_{i=1}^n\rho_{n,i}\leq \sum_{i=2}^n\E[\rho_{n,i}]/2 \}$.
On the set where this holds we have, using Theorem~\ref{thm:log-MH} and Theorem~\ref{thm:CLT}
\begin{align*}
\E\left[(1\wedge e)\left(\sum_{i=1}^n\rho_{n,i}\right)1_{\mathcal{A}_n}\right]
&\leq 
\exp\left( \frac{1}{2}\sum_{i=1}^n\E[\rho_{n,i}] \right)
\\
&\leq 
\exp\left( -\frac{n}{4}\E[\rho^2_{n,1}]\left(1-\frac{\E[\rho_{n,1}]+\frac{1}{2}\E[\rho^2_{n,1}]}{\E[\rho^2_{n,1}]}\right) \right)
\\
&\leq 
\exp\left( -\frac{n}{4}\sigma_n^6\left(\sigma^{-6}_n\E[\rho^2_{n,1}]\right)\left(1-\frac{\E[\rho_{n,1}]+\frac{1}{2}\E[\rho^2_{n,1}]}{\E[\rho^2_{n,1}]}\right) \right)
\\
&\leq 
\exp\left( -\frac{\theta^2}{8}n\sigma_n^6 \right)
\end{align*}
for all large enough $n\in\N$.  On the complement $\mathcal{A}_n^c$ using Markov's inequality gives
\begin{align*}
\E\left[(1\wedge e)\left(\sum_{i=1}^n\rho_{n,i}\right)1_{\mathcal{A}^c_n}\right]
&\leq 
\mathbb{P}\left[\mathcal{A}_n^c\right]
\\
&=  
\mathbb{P}\left[\sum_{i=1}^n\left(\rho_{n,i}-\E[\rho_{n,i}]\right)>-\frac{1}{2}\sum_{i=1}^n\E[\rho_{n,i}]\right]
\\
&\leq 
\frac{4\text{Var}[\rho_{n,1}^2]}{n\E[\rho_{n,1}]^2}
\\
&\leq 
\frac{16}{n\E[\rho_{n,1}^2]\left(1-4\frac{\E[\rho_{n,1}]+\frac{1}{2}\E[\rho_{n,1}^2]}{\E[\rho_{n,1}^2]}\right)}
\\
&\leq 
\frac{32}{\theta^2n\sigma_n^6}
\end{align*}
for all large enough $n\in\N$.

Together with \eqref{eq:aux1} and \eqref{eq:aux2} these bounds imply
\[
\E\left[(Y_{n,1}-X_{n,1})^2\left(1\wedge e\right)\left(\sum_{i=1}^n\rho_{n,i}\right)\right]
 \leq   
2\sigma_n^2\left(e^{-\frac{\theta^2}{8}n\sigma_n^6}+\frac{32}{\theta^2n\sigma_n^6}\right)
\leq 
C\frac{1}{n^{1/3}}\frac{1}{n^{2/3}\sigma_n^4}
\]
for an appropriate constant $C>0$. Since $n\sigma_n^6\to\infty$, this decay rate is faster than $n^{-1/3}$, meaning $\mathcal{E}_n$ also decays faster than $n^{-1/3}$.

Now assume, $\lim_{n\to\infty}n^{1/6}\sigma_{n,\ell}=\ell $. Splitting the expectation gives
\begin{multline*}
    \E\left[(Y_{n,1}-X_{n,1})^2\left(1\wedge e\right)\left(\sum_{i=1}^n\rho_{n,i}\right)\right]
     = 
    \E\left[(Y_{n,1}-X_{n,1})^2\right]\left[\left(1\wedge e\right)\left(\sum_{i=2}^n\rho_{n,i}\right)\right]
    \\+ 
    \E\left[(Y_{n,1}-X_{n,1})^2\left(\left(1\wedge e\right)\left(\sum_{i=1}^n\rho_{n,i}\right)-\left(1\wedge e\right)\left(\sum_{i=2}^n\rho_{n,i}\right)\right)\right]\,.
\end{multline*}
For the first term on the right-hand side we have by \eqref{eq:aux1} and Theorem~\ref{thm:CLT} that
\[
n^{1/3}\E\left[(Y^{(n)}_1-X_1)^2\right]\left[\left(1\wedge e\right)\left(\sum_{i=2}^n\rho_{\sigma_{n,\ell},i}\right)\right]
 \to
h(\ell)
\]
as $n\to\infty$.  The second term vanishes. To see this note that by Theorem~\ref{thm:CLT} and the fact that the function $t\mapsto 1\wedge e^t$ is $1$-Lipschitz and bounded imply that the expression
\[
    \left(\left(1\wedge e\right)\left(\sum_{i=1}^n\rho_{\sigma_{n,\ell},i}\right)-\left(1\wedge e\right)\left(\sum_{i=2}^n\rho_{\sigma_{n,\ell},i}\right)\right)
\]
converges to zero in the $L^2$ sense, and therefore also in probability. On the other hand \eqref{eq:aux1} and Assumption~\ref{ass1} iii) imply that the random variables $n^{1/3}(Y_{n,1}-X_{n,1})^2$ have an integrable dominating bound. Their product is therefore uniformly integrable and converges to zero in probability, meaning it is also true that
\[
n^{1/3}\E\left[(Y_{n,1}-X_{n,1})^2\left(\left(1\wedge e\right)\left(\sum_{i=1}^n\rho_{n,i}\right)-\left(1\wedge e\right)\left(\sum_{i=2}^n\rho_{n,i}\right)\right)\right]
 \to
0
\]
as $n\to\infty$.  

Finally we will optimize over the choice of $\ell$. The function $h(\ell)=2\ell^2\Phi(-\ell^3\theta/2)$ is smooth in $\ell$ and converges to zero both when $\ell\to 0$ and when $\ell \to \infty$. Hence, its maximum is attained at a stationary point. Setting $s=\ell^3\theta/2$ we can find the stationary points of $s\mapsto 2\left(\frac{2}{\theta}\right)^{2/3}s^{2/3}\Phi(-s)$. This corresponds to finding the solution of the equation $2/3 = s\varphi_N(-s)/\Phi(-s)$, where $\varphi_N$ denotes the standard Gaussian probability density function. There exists a unique solution $s^*$, as the function $s\mapsto s\varphi_N(-s)/\Phi(-s)$ is strictly increasing, meaning  $h(\ell)$ attains its maximal value at a specific value $\ell^*$ satisfying  $s^*=(\ell^*)^3\theta/2$, corresponding to an average acceptance rate of $2\Phi(-s^*)$, which turns out to numerically equal to $57.4\%$ to three decimal places. This is also implies that
\[
h(\ell^*)
 = 
\theta^{-2/3}\cdot 2^{5/3}(s^*)^{2/3}\Phi(-s^*)
 \approx 
0.651637\times \theta^{-2/3},
\]
which completes the proof.
\end{proof}

%\subsection{Section \ref{sec:optimal}}

% \begin{proof}[of Proposition \ref{prop:existence}]
% Consider family of functions $g_\lambda(t)=\frac{t^{1+\lambda}+t^{1-\lambda}}{t^{1+\lambda}+t^{-\lambda}}$ parametrized by $\lambda\in\R$. It is easy to verify that $g_\lambda$ is balancing, satisfies $g''(1)=-\lambda-\frac{1}{2}$ and that the first three derivatives of the associated $b_\lambda$ are all bounded and Assumption~\ref{ass1} ii) is satisfied. This resolves the role of the balancing function.

% Turn our attention to the density $k$. Since it has unit variance, and since $\mu_4^2\leq \mu_2\mu_6$ we must have
% \[
% 1
%  \leq 
% \mu_4^2
%  \leq 
% \mu_6\,,
% \]
% by Jensen's and Cauchy's inequalities. Furthermore if $\mu_4=1$, then $\E[(W^2-1)^2]=0$ meaning $W^2=1$ almost surely, hence $W^6=1$ almost surely and $\mu_6=1$. Therefore $(a,b)\in S$ is necessary.

% Establishing that it is sufficient can be understood as a simple case of the Hamburger moment problem \jure{(citation!)}. For clarity we still give a different direct proof...
% \end{proof}

\section{Technical results}

\begin{lemma}\label{lem:de-exponentialisation}
Let $\rho$ be a log-Metropolis--Hastings random variable associated with a probability measure $\pi$ and a Markov kernel $Q$ on $(\mathbb{X},\mathcal{F})$. Let $f: \mathbb{X}\to\R$ and $h: \mathbb{X}\times \mathbb{X}\to\R$ be such that $h(X,Y)=h(Y,X)$. Then the following are true:
\begin{enumerate}[(i)]
    \item $-\rho(X,Y)=\rho(Y,X)\,.$
    \item If the integrals are finite, then setting $h=h(X,Y)$ and $\rho = \rho(X,Y)$
\[
\E\left[h f(\rho e^{\rho})\right]
=
\E\left[h f(-\rho)\right]\,.
\]
\end{enumerate}
\end{lemma}

\begin{proof}[Proof of Lemma \ref{lem:de-exponentialisation}]
The proof is given Proposition~3 of \cite{vogrinc2021counterexamples}, with the only minor difference that we carry a symmetric function $h$ through the entire derivation.
\end{proof}

%% Bibliography
\bibliographystyle{plain}
\bibliography{paper-ref}

\end{document}